\documentclass[a4paper,UKenglish,cleveref, autoref, thm-restate]{lipics-v2021}

\pdfoutput=1 
\hideLIPIcs  

\usepackage{xcolor}
\usepackage{soul}
\newcommand{\changed}[2]{#2}

\usepackage{graphicx}
\usepackage{todonotes}

\usepackage[frozencache]{minted}

\usepackage{marvosym,amssymb,mathrsfs,mathtools,bussproofs,stmaryrd,quiver,framed}

\newcommand{\pn}{\mathsf{pn}}
\newcommand{\fv}{\mathsf{fv}}

\newcommand{\True}{\mathsf{true}}

\newcommand{\vx}{\overline}
\newcommand{\proc}[1]{\mathsf{#1}}
\newcommand{\val}[1]{\emph{#1}}

\newcommand{\Disjoint}[2]{#1\,\#\,#2}

\newcommand{\BCom}[4]{#1.#2 \to \mathsf{val}\ #3.#4 } 
\newcommand{\BRecv}[3]{#1 \rightsquigarrow #2.#3}
\newcommand{\BRecvSelect}[3]{#1 \rightsquigarrow #2[#3]}
\newcommand{\BSelect}[3]{#1 \to #2[#3]}
\newcommand{\BPure}[3]{\mathsf{val}\ #2.#1\ \mathsf{=}\ #3}
\newcommand{\BIf}[4]{\mathsf{if}\,{\At{#1}{#2}}\,\mathsf{then}\,{#3}\,\mathsf{else}\,{#4}}
\newcommand{\BCall}[2]{#1(#2)}
\newcommand{\BCalling}[4]{#1.\, #2(#3)\,\{\,#4\,\}}

\newcommand{\CCom}[5]{#5: \BCom{#1}{#2}{#3}{#4} } 
\newcommand{\CRecv}[4]{#4: \BRecv{#1}{#2}{#3}}
\newcommand{\CSelect}[4]{#4: \BSelect{#1}{#2}{#3}}
\newcommand{\CRecvSelect}[4]{#4: \BRecvSelect{#1}{#2}{#3}}
\newcommand{\CPure}[4]{#4: \BPure{#1}{#2}{#3}}

\newcommand{\CIf}[5]{#5: \BIf{#1}{#2}{#3}{#4}}
\newcommand{\CCall}[3]{#3: \BCall{#1}{#2}}
\newcommand{\CCalling}[5]{#5: \BCalling{#1}{#2}{#3}{#4}}
\newcommand{\At}[2]{#1 \MVAt #2}

\newcommand{\Tok}{\mathsf{t}}
\newcommand{\Line}{l}
\newcommand{\p}{\mathsf{p}}
\newcommand{\q}{\mathsf{q}}
\renewcommand{\r}{\mathsf{r}}
\newcommand{\NextToken}{\mathsf{nextToken}}
\newcommand{\TokenType}{\mathsf{Token}}
\newcommand{\CIfM}[3]{#3: \mathsf{if}\,{\At{#1}{#2}}}
\newcommand{\CThen}{\,\mathsf{then}\,}
\newcommand{\CElse}{\,\mathsf{else}\,}
\newcommand{\Block}[1]{\{\,#1\,\}}

\newcommand{\CCfg}[3]{\big\langle\, #1,\ #2,\ #3 \,\big\rangle}
\newcommand{\CStep}[1]{\xrightarrow{\,#1\,}}
\newcommand{\CEval}[3]{#1 \vdash #2 \Downarrow #3}

\newcommand{\PRecv}[3]{?_{#2,#3}\,#1}
\newcommand{\PSet}[2]{\mathsf{val}\ {#1}\ \mathsf{=}\ {#2}}

\newcommand{\PSend}[4]{{#1}\,!_{#2,#3}\,{#4}}
\newcommand{\PSelect}[4]{{#1}\oplus_{#2,#3}{#4}}
\newcommand{\PBranch}[3]{\binampersand\{(#1) \Rightarrow #2\}_{#3}}
\newcommand{\PIf}[3]{\mathsf{if}\,{#1}\,\mathsf{then}\,{#2}\,\mathsf{else}\,{#3}}
\newcommand{\Par}{\ |\ }
\newcommand{\PCall}[4]{#3,#4: #1(#2)}

\newcommand{\PIfM}[1]{\mathsf{if}\,{#1}}
\newcommand{\PThen}{\,\mathsf{then}\,}
\newcommand{\PElse}{\,\mathsf{else}\,}

\newcommand{\BBranch}{\binampersand}

\newcommand{\PCfg}[3]{\CCfg{#1}{#2}{#3}}
\newcommand{\PStep}[1]{\CStep{#1}}

\newcommand{\EPP}[2]{\llbracket #1 \rrbracket_{#2}}
\newcommand{\Merge}{\sqcup}
\newcommand{\Extends}{\sqsupseteq}

\newcommand{\wf}{\,\checkmark}
\newcommand{\dom}{\mathsf{dom}}
\newcommand{\keys}{\mathsf{keys}}
\newcommand{\key}{\mathsf{key}}
\newcommand{\stats}{\mathsf{stats}}

\makeatletter
\newenvironment{myalign*}{\ifvmode\else\hfil\null\linebreak\fi
  \hspace*{-\leftmargin}\minipage\textwidth
  \setlength{\abovedisplayskip}{0pt}%
  \setlength{\abovedisplayshortskip}{\abovedisplayskip}%
  \start@align\@ne\st@rredtrue\m@ne}%
{\endalign\endminipage\linebreak}
\makeatother

\title{Ozone: Fully Out-of-Order Choreographies}

\author{Dan Plyukhin}{University of Southern Denmark, Denmark }{dplyukhin@imada.sdu.dk}{https://orcid.org/0009-0004-8712-7895}{}

\author{Marco Peressotti}{University of Southern Denmark, Denmark}{peressotti@imada.sdu.dk}{https://orcid.org/0000-0002-0243-0480}{}
\author{Fabrizio Montesi}{University of Southern Denmark, Denmark}{fmontesi@imada.sdu.dk}{https://orcid.org/0000-0003-4666-901X}{}

\authorrunning{D. Plyukhin, M. Peressotti, and F. Montesi} 

\Copyright{Dan Plyukhin, Marco Peressotti, and Fabrizio Montesi} 



\supplement{}
\supplementdetails[subcategory={Source Code}, cite={}, swhid={}]{Software}{https://github.com/dplyukhin/ozone}

\funding{Partially supported by Villum Fonden (grant no. 29518). Co-funded by the European Union (ERC, CHORDS, 101124225). Views and opinions expressed are however those of the authors only and do not necessarily reflect those of the European Union or the European Research Council. Neither the European Union nor the granting authority can be held responsible for them.}

\nolinenumbers 

\EventEditors{Jonathan Aldrich and Guido Salvaneschi}
\EventNoEds{2}
\EventLongTitle{38th European Conference on Object-Oriented Programming (ECOOP 2024)}
\EventShortTitle{ECOOP 2024}
\EventAcronym{ECOOP}
\EventYear{2024}
\EventDate{September 16--20, 2024}
\EventLocation{Vienna, Austria}
\EventLogo{}
\SeriesVolume{313}
\ArticleNo{44}

\begin{document}

\maketitle

\begin{abstract}
    Choreographic programming is a paradigm for writing distributed applications. It allows programmers to write a single program, called a choreography, that can be compiled to generate correct implementations of each process in the application. Although choreographies provide good static guarantees, they can exhibit high latency when messages or processes are delayed. This is because processes in a choreography typically execute in a fixed, deterministic order, and cannot adapt to the order that messages arrive at runtime. In non-choreographic code, programmers can address this problem by allowing processes to execute out of order---for instance by using futures or reactive programming. However, in choreographic code, out-of-order process execution can lead to serious and subtle bugs, called \emph{communication integrity violations (CIVs)}.

    In this paper, we develop a model of choreographic programming for out-of-order processes that guarantees absence of CIVs and deadlocks. As an application of our approach, we also introduce an API for safe non-blocking communication via futures in the choreographic programming language Choral. The API allows processes to execute out of order, participate in multiple choreographies concurrently, and to handle unordered data messages. We provide an illustrative evaluation of our API, showing that out-of-order execution can reduce latency and increase throughput by overlapping communication with computation.
\end{abstract}

\section{Introduction}

Choreographic programming~\cite{montesi2013phd} is a paradigm that simplifies writing distributed applications. In contrast to a traditional development style, where one implements a separate program for each type of process in the system, choreographic programming allows a programmer to define the behaviors of all processes together in a single program called a \emph{choreography}~\cite{montesi2023}. Through \emph{endpoint projection (EPP)}, a choreography can be compiled to generate the programs implementing each process that would otherwise need to be written by hand. Aside from convenience, the advantage of this approach is that certain classes of bugs (such as deadlocks) are impossible \emph{by construction}~\cite{carbone2013}. Choreographic programming has been applied to popular languages such as Java~\cite{giallorenzo2023} and Haskell~\cite{shen2023a}, and has been used to implement real-world protocols such as IRC~\cite{lugovic2023}.

\begin{figure}
    \centering
    \begin{subfigure}{0.34\textwidth} 
        \small
\begin{myalign*}
&1: \BCom{\p_1}{produce()}{\q}{x_1};\\
&2: \BCom{\p_2}{produce()}{\q}{x_2};\\
&3: \BCom{\q}{compute(\q.x_1)}{\p_1}{y_1};\\
&4: \BCom{\q}{compute(\q.x_2)}{\p_2}{y_2}
\end{myalign*}

        \subcaption{Choreography}\label{fig:concurrent-producers-chor}
    \end{subfigure}
    \begin{subfigure}{0.3\textwidth}
        \includegraphics[width=0.60\textwidth]{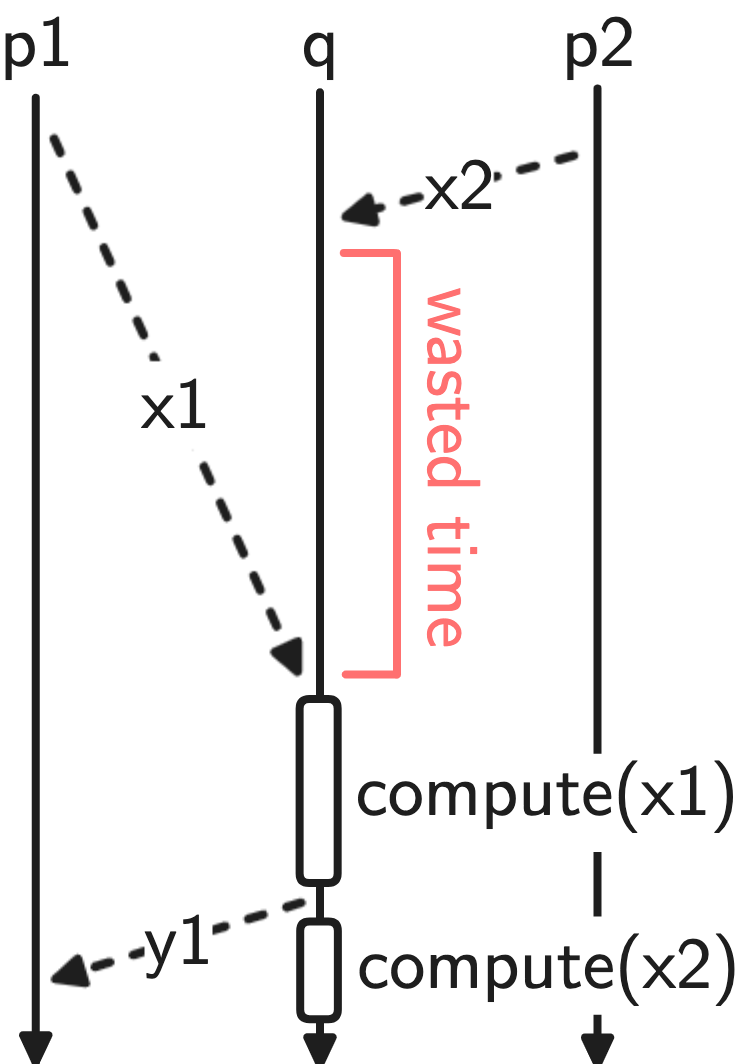}

        \subcaption{In-order execution}\label{fig:concurrent-producers-actual}
    \end{subfigure}
    \begin{subfigure}{0.3\textwidth}
        \includegraphics[width=0.60\textwidth]{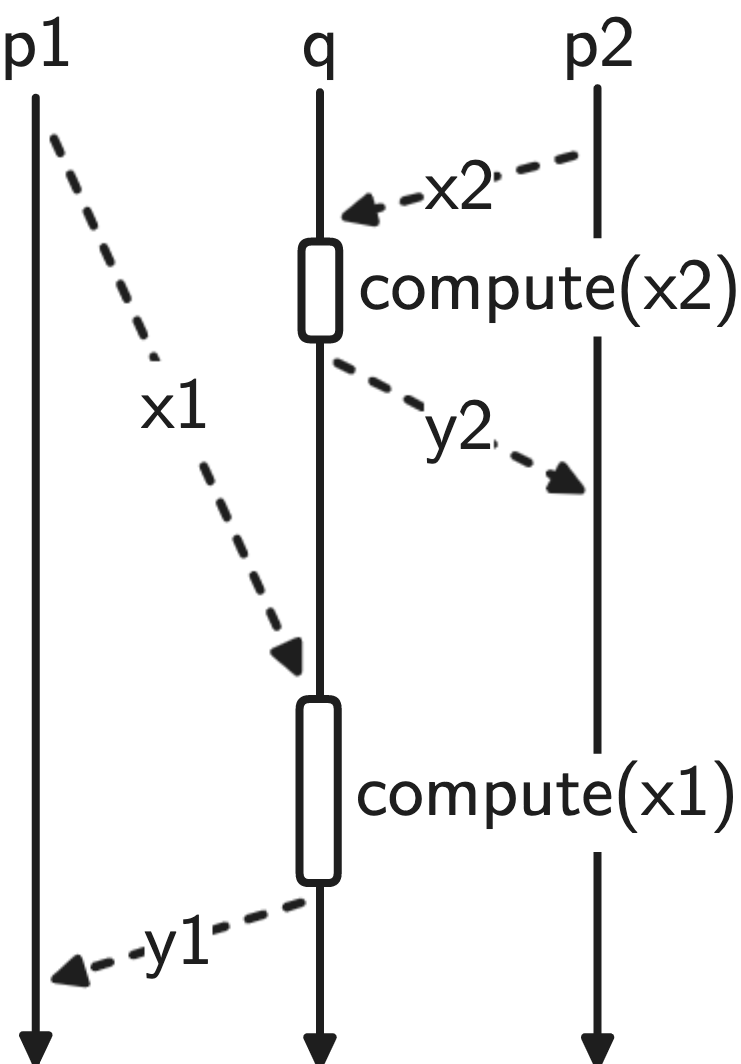}

        \subcaption{Out-of-order execution}\label{fig:concurrent-producers-ideal}
    \end{subfigure}

    \caption{A choreography where out-of-order execution can improve performance.}%
    \label{fig:concurrent-producers}
\end{figure}


Processes in choreographic programs typically execute in a fixed, sequential order. Consider \Cref{fig:concurrent-producers-chor}, which shows a simple choreography performed by processes $\p_1$, $\p_2$, and $\q$. The syntax $\BCom{\p}{e}{\q}{x}$ means ``$\p$ evaluates expression $e$ and sends the result to $\q$, which binds the result to a local variable $x$''. Under the usual semantics for choreographies, $\p_1.produce()$ and $\p_2.produce()$ can be evaluated in parallel because $\p_1$ and $\p_2$ are distinct processes. However, $\q$ must execute each step sequentially: first $\q$ waits until it receives $x_1$; then $\q$ waits until it receives $x_2$; and only then can $\q$ send $\p_1$ the result of processing $x_1$.

\Cref{fig:concurrent-producers-actual} depicts an execution of the choreography, showing the drawback of a fixed processing order: if $x_2$ arrives before $x_1$, $\q$ wastes time waiting for $x_1$ instead of processing $x_2$. Ideally, $\q$ would evaluate $compute(\q.x_1)$ and $compute(\q.x_2)$ according to the arrival order of $x_1$ and $x_2$, as shown in \Cref{fig:concurrent-producers-ideal}. Assuming these two expressions are safe to reorder, such an optimization would allow $\q$ to overlap computation with communication and reduce the average latency experienced by $\p_1$ and $\p_2$. We are therefore interested in studying choreographic programming models where processes may execute some statements out of order, or even concurrently. We call such processes \emph{out-of-order processes} and the corresponding choreographies \emph{(fully) out-of-order choreographies}.


Processes with out-of-order features have been considered in prior work. Process models such as the actor model~\cite{agha1990a} or the $\pi$-calculus with delayed receive~\cite{merro2004} are expressive enough to implement the behavior in \Cref{fig:concurrent-producers-ideal}, but these models lack the static guarantees of choreographic programming. More recently, Montesi gave a semantics for \emph{nondeterministic choreographies}~\cite{montesi2023}, i.e., choreographies with nondeterministic choice. Nondeterministic choreographies can implement the execution in \Cref{fig:concurrent-producers-ideal}, but they are unwieldy when it comes to expressing out-of-order process execution: they require explicitly writing all possible schedulings, lest getting a suboptimal program. For our example, we would get a choreography twice the size of the one in \Cref{fig:concurrent-producers-chor} (cf.~\Cref{sec:related-work}). Consequently, nondeterministic choreographies are both hard to write and brittle---a typical drawback when using syntactic operators to express interleavings. This raises the question:

\begin{quote}
    \emph{Can we develop a choreographic programming model for out-of-order processes that marries the simple syntax of \Cref{fig:concurrent-producers-chor} with the semantics of \Cref{fig:concurrent-producers-ideal}?}
\end{quote}
The simplicity of this problem is deceptive, since common-sense approaches can lead to pernicious compiler bugs. For instance, consider \Cref{fig:forwarding-data}: two microservices $\proc {cs}$ (a ``content service'') and $\proc {ks}$ (a ``key service'') send values $txt, key$ to a server $\proc s$ (lines 1 and 2). The server in turn forwards those values to a client $\proc c$ (lines 3 and 4). Notice that if $\proc s$ is an out-of-order process, then it can forward the results in any order, as shown in \Cref{fig:forwarding-data-1,fig:forwarding-data-2}. This causes a problem for $\proc c$: since both $txt$ and $key$ were sent by $\proc s$, and since both values have the same type ($\mathsf{String}$), $\proc c$ has no way to determine whether the first message contains $txt$ (as in \Cref{fig:forwarding-data-1}) or $key$ (as in \Cref{fig:forwarding-data-2}). This problem is easy for compiler writers to miss, leading to disastrous nondeterministic bugs where variables are bound to the wrong values. We call such bugs \emph{communication integrity violations (CIVs)}.

In this paper, we investigate CIVs and other complications that arise from mixing choreographies with out-of-order processes. Our investigation brings forward necessary elements that are missing from previous research on choreographic programming~\cite{montesi2023} and the neighbouring approach of multiparty session types (which use simpler choreographies without data or computation)~\cite{honda2016,DH12,APN17,VHEZ21}.  Although the problem in \Cref{fig:forwarding-data} can easily be solved by attaching static information (such as variable names) to each message, we show in \Cref{sec:motivation} that a general solution requires mixing static and dynamic information, replicated across multiple processes. We also find that formalizing fully out-of-order choreographies requires several features uncommon in standard choreographic programming models, such as scoped variables and an expanded notion of well-formedness.

We make the following key contributions:
\begin{enumerate}
    \item We present $O_3$, a formal model for asynchronous, fully out-of-order choreographies.\footnote{The name $O_{3}$ derives from our model being {\it O\/}ut {\it O\/}f {\it O\/}rder.} Our model prevents CIVs by attaching \emph{integrity keys} to messages. A nice consequence of our solution is that messages no longer need to be delivered in FIFO order. We prove that $O_3$ choreographies ensure deadlock-freedom (\Cref{thm:progress}) and communication integrity (\Cref{thm:integrity}).
    \item We present an EPP algorithm to project $O_3$ choreographies into out-of-order processes. We prove an operational correspondence theorem, which states that a choreography and its projection evolve in lock-step (\Cref{thm:epp}). The key to making this proof tractable is a new notion of well-formedness that formalizes a communication integrity invariant. The theorem implies that a correct compiler will not generate code with deadlocks or CIVs.
    \item We demonstrate the applicability of our approach by developing \emph{Ozone}, a non-blocking communication API for the choreographic programming language Choral~\cite{giallorenzo2023}.\footnote{The name \emph{Ozone} derives from $O_{3}$ being the chemical formula for ozone.} Choreographies implemented with Ozone can use futures~\cite{baker1977} to process messages concurrently (as in \Cref{fig:concurrent-producers-ideal}) without violating communication integrity.
    We evaluate Ozone with microbenchmarks and a model serving pipeline~\cite{wang21}.
    Our results confirm that out-of-order execution can dramatically reduce latency and increase throughput for choreographies, putting the performance of hand-written reactive processes within reach of choreographic programmers (we compare to actors written in the popular \emph{Akka} framework~\cite{akka}).
\end{enumerate}

The paper is structured as follows. \Cref{sec:motivation} explores CIVs and other issues in out-of-order choreography models. \Cref{sec:choreographies} presents our formal model $O_3$. \Cref{sec:processes} presents our model for out-of-order processes and our EPP. \Cref{sec:choral} presents our non-blocking API for Choral and our evaluation. We conclude with related work in \Cref{sec:related-work} and discussion in \Cref{sec:conclusion}.

\begin{figure}
    \begin{subfigure}{0.33\textwidth}  
        \small
\begin{myalign*}
&1: \BCom{\proc {cs}}{getText()}{\proc s}{txt};\\
&2: \BCom{\proc {ks}}{getKey()}{\proc s}{key};\\
&3: \BCom{\proc s}{txt}{\proc c}{txt};\\
&4: \BCom{\proc s}{key}{\proc c}{key};\\
&5: \proc{c}.display(txt);\\
&6: \proc{c}.decrypt(key)
\end{myalign*}
        \caption{Choreography}\label{fig:forwarding-data-chor}
    \end{subfigure}\hfill
    \begin{subfigure}{0.33\textwidth}
        \includegraphics[width=\textwidth]{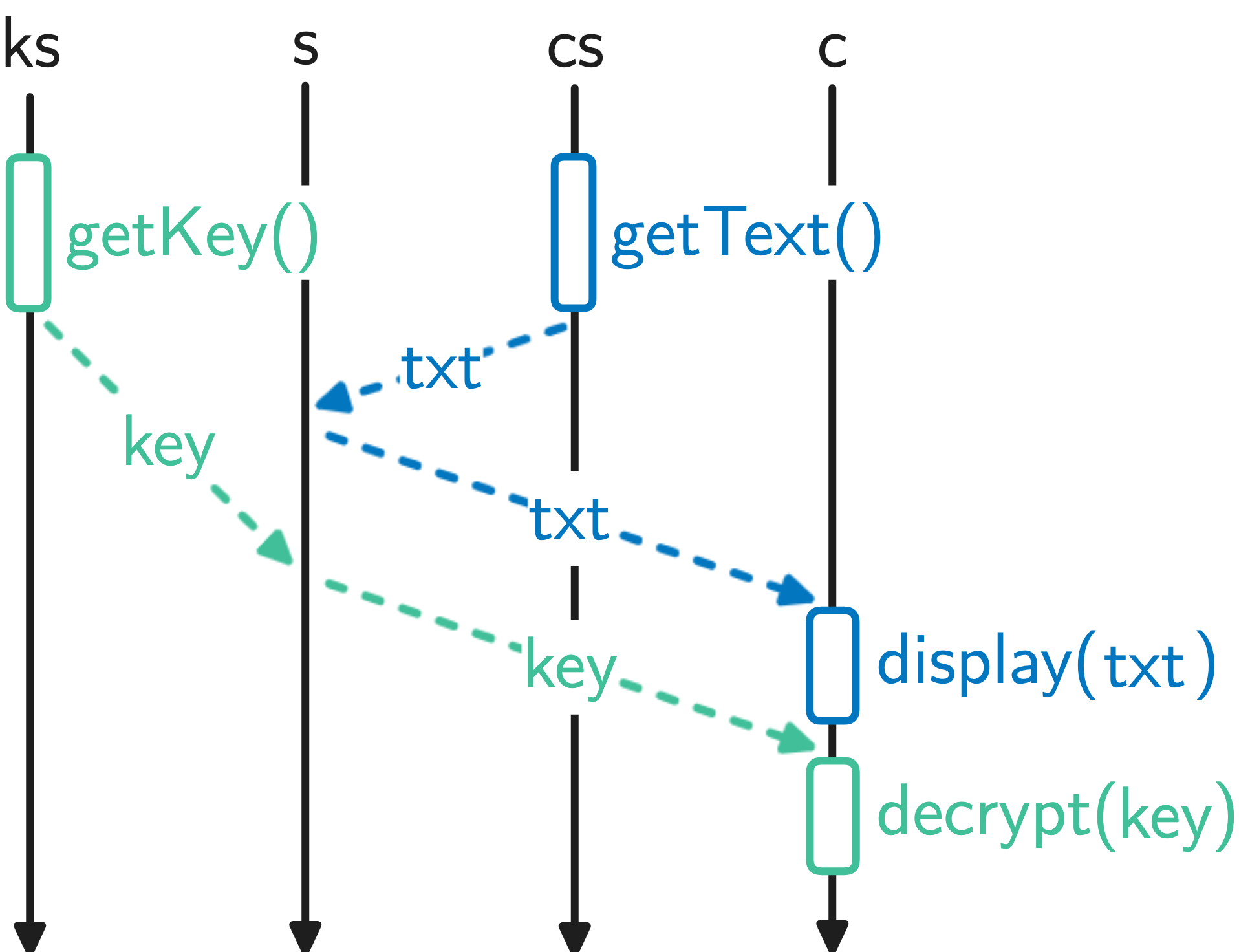}
        \caption{Safe execution}\label{fig:forwarding-data-1}
    \end{subfigure}\hfill
    \begin{subfigure}{0.33\textwidth}
        \includegraphics[width=\textwidth]{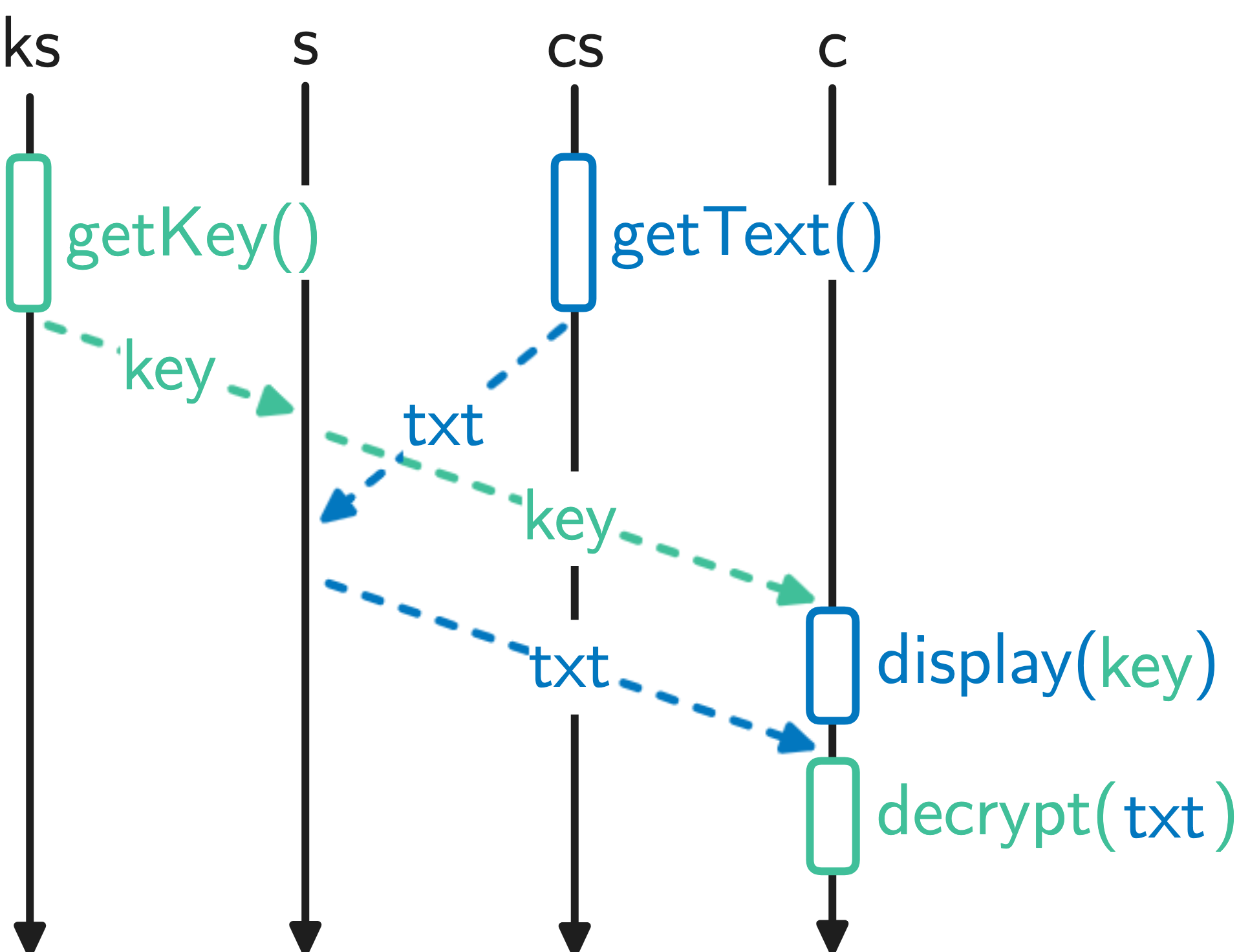}
        \caption{Buggy execution}\label{fig:forwarding-data-2}
    \end{subfigure}
    \caption{A choreography where na\"ive out-of-order execution is unsafe. Process $\proc c$ cannot distinguish whether the first message it receives represents \emph{key} or \emph{txt}.}
    \label{fig:forwarding-data}
\end{figure}

\section{Overview}\label{sec:motivation}

In this section we explore the challenges that must be solved to develop a fully out-of-order choreography model, along with our approach.

\subsection{Intraprocedural Integrity}\label{sec:intraprocedural-civs}

Informally, communication integrity is the property that messages communicated in a choreography are bound to the correct variables. To ensure this, processes might need extra information; in \Cref{fig:forwarding-data}, process $\proc c$ needs to know which value will arrive first: $txt$ or $key$.

A traditional solution would be for $\proc s$ to send a \emph{selection} to $\proc c$. Selections are communications of constant values, used in choreography languages when one process makes a control flow decision that other processes must follow. \Cref{fig:intraprocedural-disambiguation-selection} shows how $\proc s$ could send the selection $\mathsf{[KEY]}$ to inform $\proc c$ that $key$ will arrive before $txt$. Indeed, this is the approach used by nondeterministic choreographies~\cite{montesi2023}. However, selections impose overhead: any time nondeterminism could occur, the programmer would need to insert new selection messages. These extra messages would have both a cognitive cost for the programmer (as programs become littered with selections) and a runtime cost in the form of an extra message.

\begin{figure}
    \centering
    \begin{subfigure}{0.35\textwidth}
        \includegraphics[width=\textwidth]{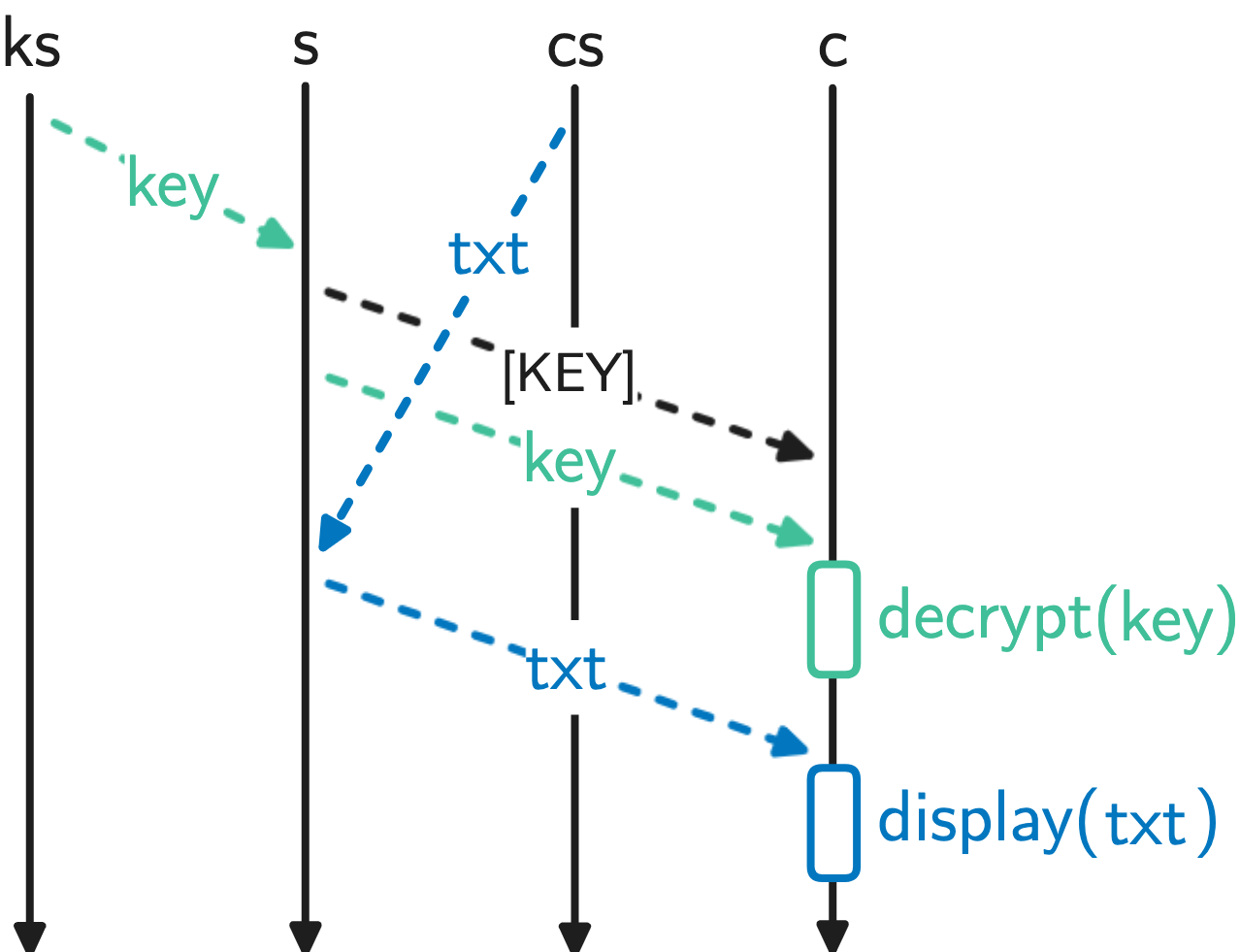}
        \caption{Using selections}\label{fig:intraprocedural-disambiguation-selection}
    \end{subfigure}
    \begin{subfigure}{0.35\textwidth}
        \includegraphics[width=\textwidth]{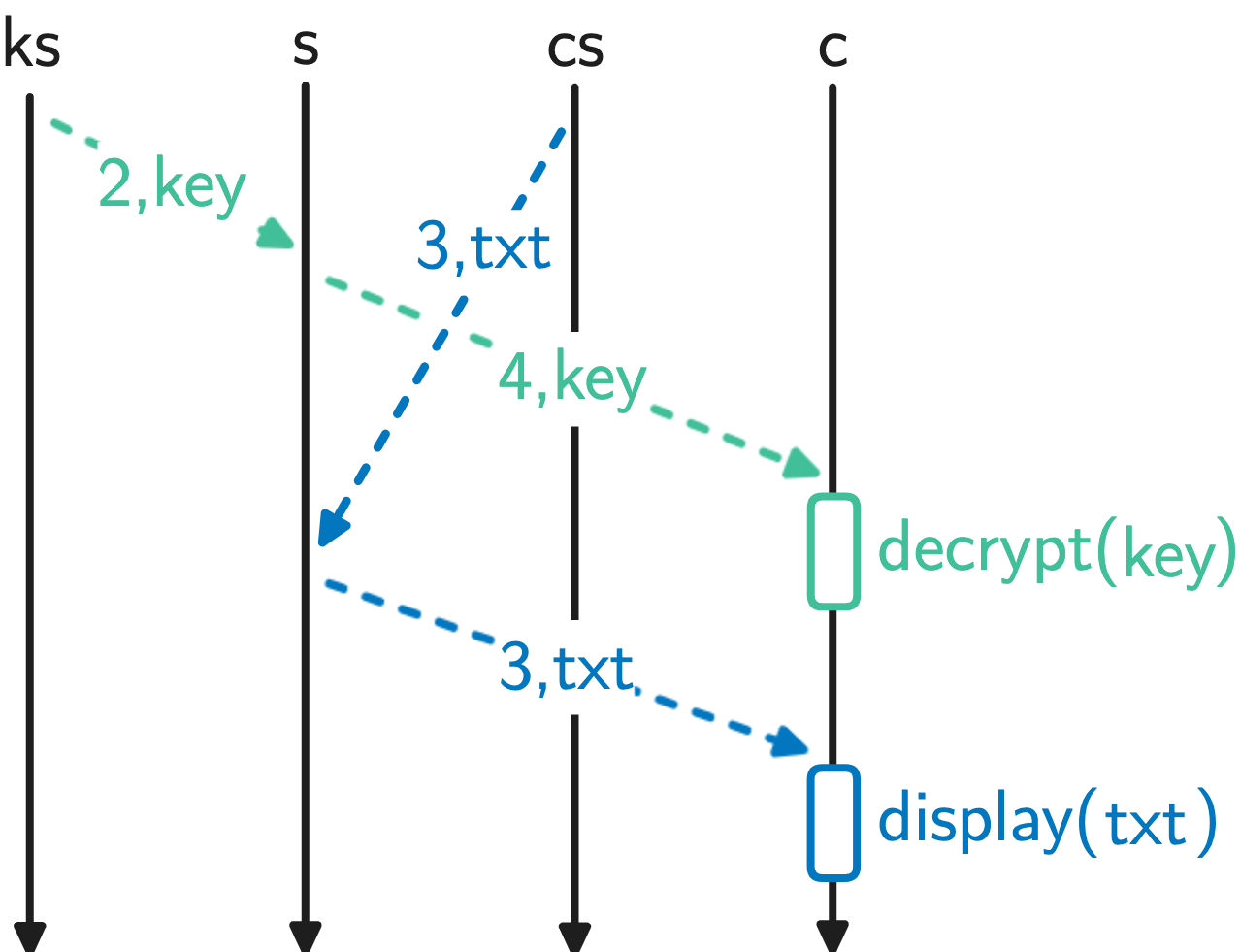}
        \caption{Using integrity keys}\label{fig:intraprocedural-disambiguation-keys}
    \end{subfigure}
    \caption{Two approaches to prevent CIVs: selections and integrity keys.}
    \label{fig:intraprocedural-disambiguation-options}
\end{figure}

\begin{figure}
    \centering
    \begin{subfigure}{0.22\textwidth}
        \includegraphics[width=0.9\textwidth]{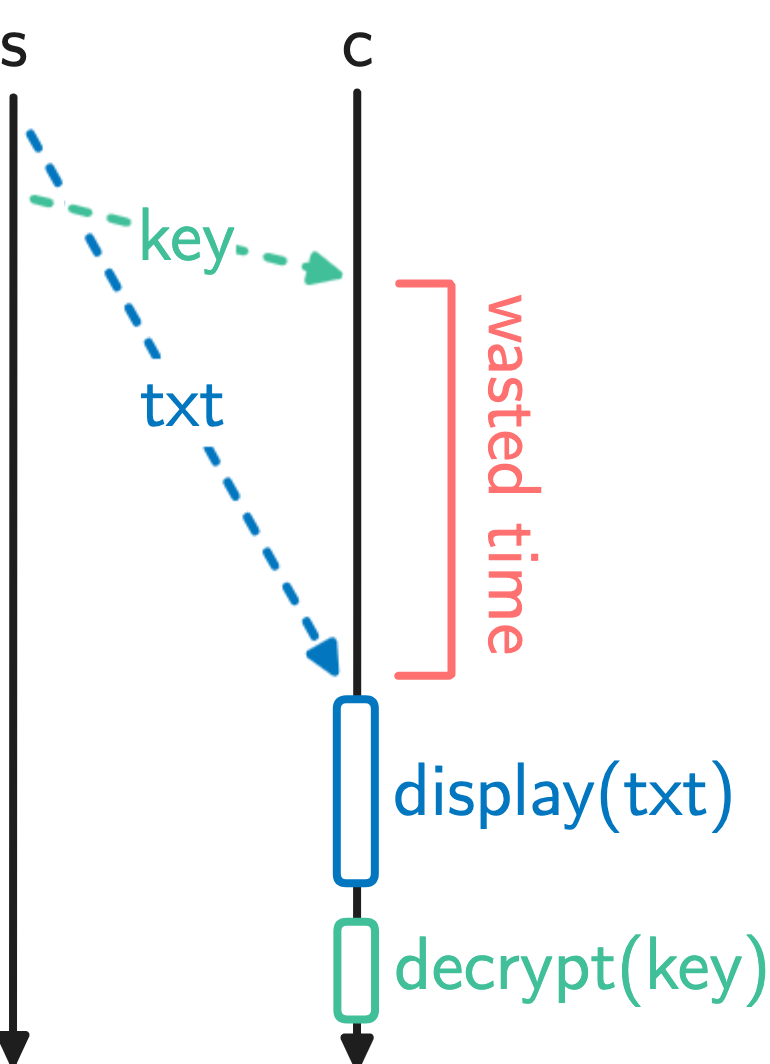}
        \caption{}\label{fig:head-of-line-blocking}
    \end{subfigure}\qquad
    \begin{subfigure}{0.22\textwidth}
        \includegraphics[width=0.9\textwidth]{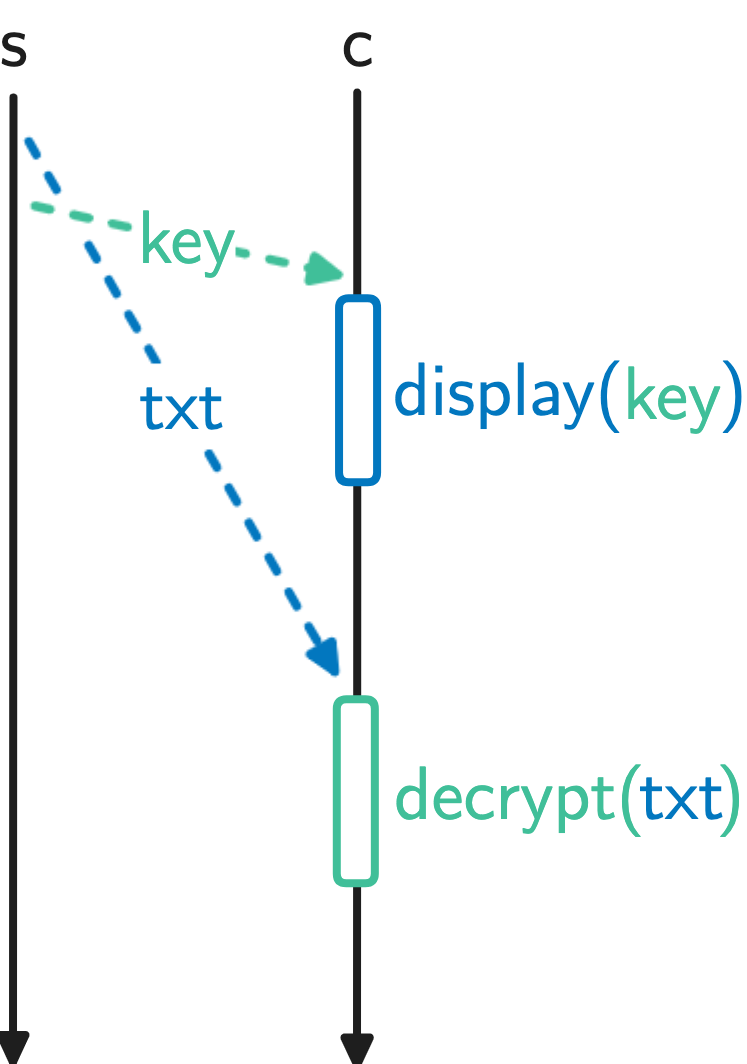}
        \caption{}\label{fig:unordered-civ}
    \end{subfigure}\qquad
    \begin{subfigure}{0.22\textwidth}
        \includegraphics[width=0.9\textwidth]{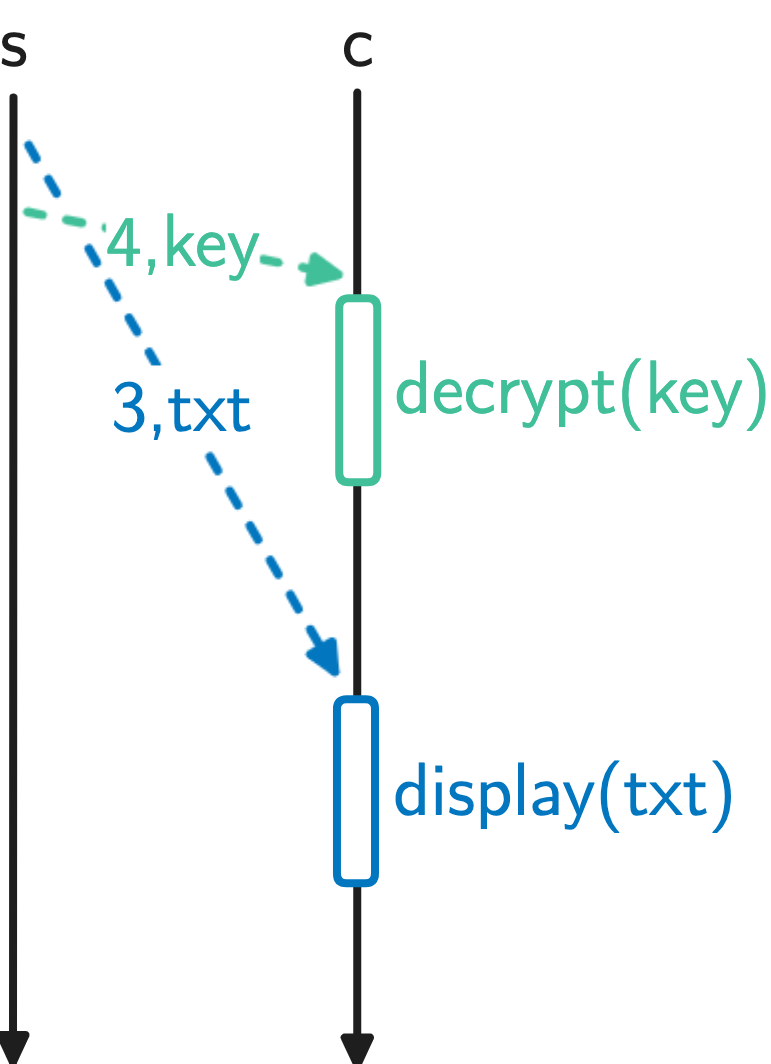}
        \caption{}\label{fig:unordered-keys}
    \end{subfigure}
    \caption{The challenges of non-FIFO delivery. Part (a) depicts head-of-line blocking when using a FIFO transport protocol: The message containing $k$ arrives first, but it cannot be processed until $t$ arrives. Part (b) depicts a CIV caused by using an unordered transport protocol without integrity keys. Part (c) depicts how the processes can use integrity keys to prevent CIVs. }\label{fig:unreliable-civs}
\end{figure}

Instead, we opt to pair each message with a disambiguating tag called an \emph{integrity key}. When $\proc c$ receives a message, it checks the integrity key to find the meaning of the message. \Cref{fig:intraprocedural-disambiguation-keys} uses \emph{line numbers} as integrity keys. For example, the $txt$ message is tagged with the number 3 because it was produced by the instruction on line 3 in \Cref{fig:forwarding-data}. Equivalently, one could use variable names (assuming that all variables have distinct names), message types (assuming that all messages have distinct types), or operators~\cite{carbone2012a}\changed{}{; essentially these are all ways to combine messages with selections}. However, as we will see in the next section, none of these solutions will suffice once we introduce procedures and recursion.

Integrity keys have another advantage over selections: they make it safe for the network to reorder messages. For instance, the selection in \Cref{fig:intraprocedural-disambiguation-selection} will only prevent CIVs if \emph{key} and \emph{txt} are delivered in the same order they were sent.
Thus previous theories and implementations of choreographic languages require a transport protocol that ensures reliable FIFO communication~\cite{giallorenzo2023,montesi2023}. These models are therefore susceptible to head-of-line blocking~\cite{scharf2006}, where one delayed message can prevent others from being processed (\Cref{fig:head-of-line-blocking}).
\Cref{fig:unordered-civ} shows why FIFO is necessary in these models: unordered messages can cause CIVs. Because our model combines unordered messages, integrity keys, and out-of-order processes, it circumvents the head-of-line blocking problem---as shown in \Cref{fig:unordered-keys}.

\subsection{Procedural Choreographies}\label{sec:procedural-choreos}

Choreographies can use procedures parameterised on processes for modularity and recursion~\cite{cruz-filipe2017,montesi2023}. \Cref{fig:proc-ambiguity} shows an example: a procedure $X$ with three \emph{roles} (i.e., process parameters) $\proc a,\proc b,\proc c$. The procedure $X$ is invoked twice---once with processes $\p,\q,\r_1$ (line 7) and again with $\p,\q,\r_2$ (line 8). In the body of $X$, role $\proc a$ produces a value and sends it to $\proc b$; then $\proc b$ transforms the value and sends it to $\proc c$; finally, $\proc c$ processes the value and sends it to $\proc a$. As usual in most programming languages, we will assume the variables $\proc a.w,\proc b.x,\proc c.y$, and $\proc a.z$ are locally scoped---this is in contrast to many choreography models, where variables at processes are all mutable fields accessible anywhere in the program.

In existing choreography models, a process can only participate in one choreographic procedure at a time. This is no longer the case with fully out-of-order choreographies. Consider \Cref{fig:proc-ambiguity}, where process $\p$ invokes procedure $X$ twice. The process may begin by invoking the first procedure call (line 7), computing $\p.w$ (line 2), and sending $\p.w$ to $\r_1$ (line 3). Then, instead of executing its next instruction---i.e.\ becoming blocked by waiting for a message on line 5---$\p$ can skip the instruction and proceed to invoke the second procedure call (line 8). Thus, we can have an execution like in \Cref{fig:proc-ambiguity-diagram-good}, in which $\p$ sends a message to $\r_1$ as part of the first procedure call and immediately sends a message to $\r_2$ as part of the second procedure call. This unusual semantics is exactly what we would expect in a choreography language with non-blocking receive---such as in Choral, when using the Ozone API to bind the result of a communication to a future (\Cref{sec:choral}).

\subsubsection{Interprocedural Integrity}\label{sec:interprocedural-civs}

Concurrent choreographic procedures add another dimension of complexity to the communication integrity problem. \Cref{fig:proc-ambiguity-diagram-bad,fig:proc-ambiguity-diagram-good} show why: depending on the order that $\r_1$ and $\r_2$'s messages arrive at $\q$, the messages from $\q$ may arrive at $\p$ in any order. (This occurs even if we assume reliable FIFO delivery!) Like in the previous section, $\p$ cannot distinguish which message pertains to which procedure invocation. But now static information is insufficient to ensure communication integrity: both messages from $\q$ pertain to the same variable in the same procedure, so the integrity keys fail to distinguish the different procedure calls. We call this the \emph{interprocedural CIV problem}.

The example above shows that integrity keys need dynamic information prevent CIVs. We can solve the problem by combining the line numbers used in \Cref{sec:intraprocedural-civs} with some \emph{session token} $t$ that uniquely identifies each procedure invocation. Applied to \Cref{fig:proc-ambiguity-diagram-bad,fig:proc-ambiguity-diagram-good}, $\p$ could inspect the session token to determine whether the messages pertain to the first procedure call (line 7) or the second (line 8). But this requires $\p$ and $\q$ to somehow achieve \emph{a priori} agreement about which tokens correspond to which procedure invocations.

One solution to the interprocedural CIV problem would be to select a ``leader'' process for each procedure call, and let the leader compute a session token for all the other roles to use. However, this would make the leader a bottleneck: until the other participants receive the token, senders would not be able to send messages, and recipients would not be able to discern which procedure invocation their incoming messages pertain to. We therefore propose a method for processes to compute session tokens independently, using only local data, such that they still agree on the same value of the token for each procedure invocation.

Observe that a procedure call is uniquely identified by its caller (i.e.~the procedure call that called it) and its line number $l$. Assuming the caller already has a unique token $t$, the callee's token can be computed as some injective function $\NextToken(l,t)$. This function would need to satisfy two properties:

\begin{itemize}
    \item \textbf{Determinism:} For any input pair $l,t$, $\NextToken(l,t)$ always produces the same value $t'$.
    \item \textbf{Injectivity:} Distinct input pairs $l,t$ produce distinct output tokens. 
\end{itemize}

\emph{Determinism} ensures that if two processes in the same procedure call (with token $t$) invoke the same procedure (on line $l$) then both processes will agree on the value of $\NextToken(l,t)$. \emph{Injectivity} ensures that if a process concurrently participates in two different procedure calls (with distinct tokens $t_1,t_2$) and invokes two procedures (on lines $l_1,l_2$---possibly $l_1 = l_2$) then the resulting session tokens will be distinct ($\NextToken(l_1,t_1) \ne \NextToken(l_2,t_2)$). In the next section, we realize these constraints by representing tokens as lists of line numbers and defining $\NextToken$ to be the list-prepend operator.

\begin{figure}
    \begin{subfigure}{0.37\textwidth}
        \small\centering
        \begin{myalign*}
            &1: X(\proc a,\proc b,\proc c) =\\
            &2: \quad \BPure{w}{\proc a}{produce()};\\
            &3: \quad \BCom{\proc a}{w}{\proc b}{x};\\
            &4: \quad \BCom{\proc b}{transform(x)}{\proc c}{y};\\
            &5: \quad \BCom{\proc c}{process(y)}{\proc a}{z};\\
            &6: \quad \proc a.store(w, z)\\
            &7: X(\p,\r_1,\q);\\
            &8: X(\p,\r_2,\q)
        \end{myalign*}
        \caption{}\label{fig:proc-ambiguity}
    \end{subfigure}
    \begin{subfigure}{0.3\textwidth}
        \includegraphics[width=0.8\textwidth]{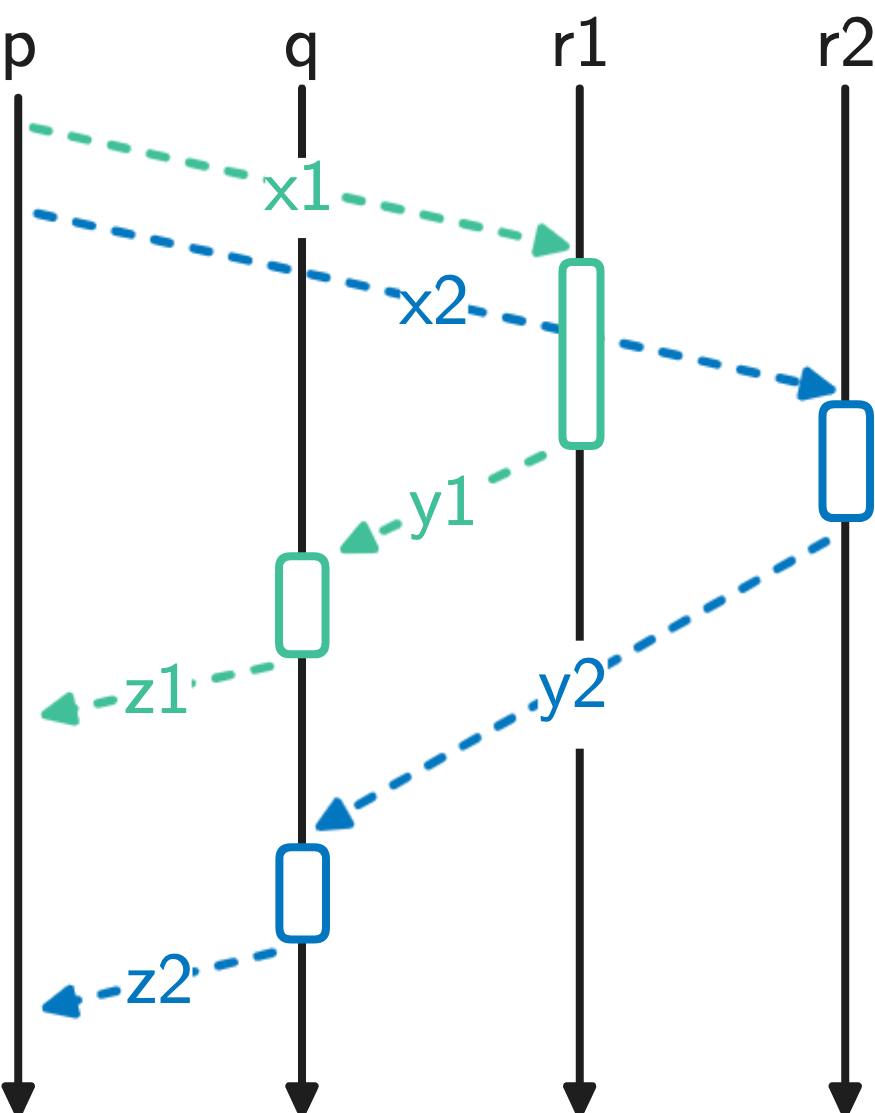}
        \caption{}\label{fig:proc-ambiguity-diagram-good}
    \end{subfigure}\hfill
    \begin{subfigure}{0.3\textwidth}
        \includegraphics[width=0.8\textwidth]{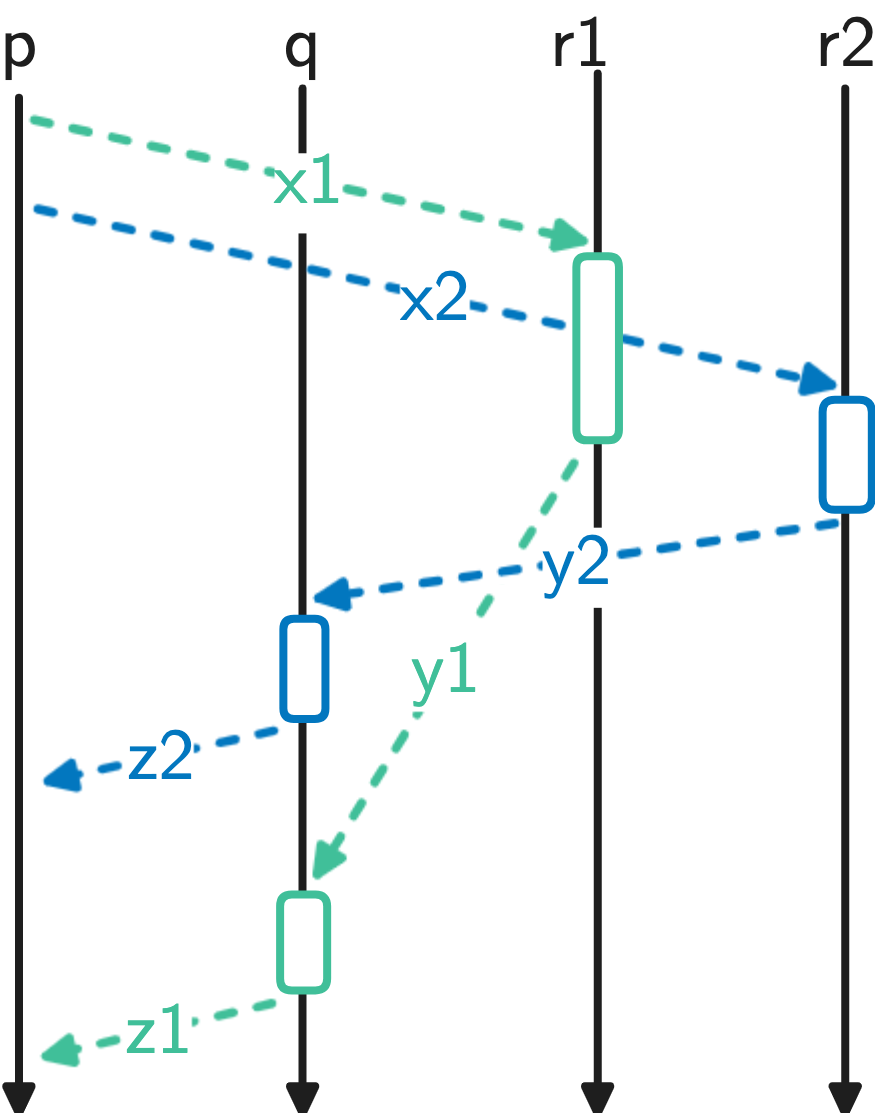}
        \caption{}\label{fig:proc-ambiguity-diagram-bad}
    \end{subfigure}
    \caption{A choreography and two possible executions. In both diagrams, the green lines correspond to $X(\p,\q,\r_1)$ and the blue lines correspond to $X(\p,\q,\r_2)$.}
    \label{fig:the-need-for-session-tokens}
\end{figure}

\section{Choreography Model}\label{sec:choreographies}

In this section we present $O_3$, a formal model for asynchronous, fully out-of-order choreographies. Statements can be executed in any order (up to data dependency) and messages can be delivered out of order. The section concludes with proofs of deadlock-freedom and communication integrity.

\subsection{Syntax}\label{sec:chor-syntax}

The syntax for choreographies in $O_3$ is defined by the grammar in \Cref{fig:chor-syntax}. Two example choreographies are shown in \Cref{fig:chor-examples}; we explain their semantics in \Cref{sec:chor-discussion}.

\newcommand{\alt}{\ |\ }
\begin{figure}
    \centering
\begin{align*}
\mathscr C &::= \{X_i(\vx \p, \vx{\p.x}) = C_i\}_{i \in \mathcal I} & \emph{(decls)} &&\\
C &::= I;\, C & \emph{(seq)} &\quad\alt \Block{C} & \emph{(block)}\\
&\alt 0& \emph{(end)} &&\\
I &::= l, t: \BCom{\p}{e}{\q}{x} & \emph{(comm)} &\quad\alt l, t: \BSelect{\p}{\q}{L} & \emph{(sel)}\\
&\alt l, t: \BPure{x}{\p}{e} & \emph{(expr)} &\quad\alt l, t: \BIf{e}{\p}{C_1}{C_2} & \emph{(cond)}\\
&\alt l, t: \BCall{X}{\vx \p,\vx a} & \emph{(call)} &\quad\alt l, t: \BRecv{\p}{\q}{x} & \emph{(comm$^\dag$)}\\
&\alt l, t: \BRecvSelect{\p}{\q}{L} & \emph{(sel$^\dag$)} &\quad\alt l, t: \BCalling{\vx \p}{X}{\vx \q,\vx a}{C} & \emph{(call$^\dag$)}\\
t &::= \Tok & \emph{(placeholder)} &\quad\alt \tau & \emph{(token$^\dag$)}\\
e &::= f(\vx e)& \emph{(app)} &\quad\alt a& \emph{(atom)}\\
a &::= \At{v}{\p}&\emph{(val)} &\quad\alt \p.x& \emph{(var)}
\end{align*}
    \caption{Syntax for choreographies in $O_3$. Terms marked with $\dag$ only appear at runtime.}%
    \label{fig:chor-syntax}
\end{figure}

A choreography $C$ is executed in the context of a collection of \emph{procedures} $\mathscr C$. Each procedure $X_i(\vx \p, \vx{\p.x}) = C_i$ is parameterized by a list of \emph{roles} $\vx \p = \p_1, \dots, \p_n$ and role-local \emph{parameters} $\vx{\p.x} = \p_{j_1}.x_1, \dots, \p_{j_m}.x_m$ where every parameter $\p_{j_k}.x_k$ is located at one of the roles in $\vx \p$. We assume that procedures do not contain runtime terms (such as $\CRecv{\p}{\q}{x}{l,t}$).

A choreography $C$ consists of a sequence of \emph{instructions} $I$, followed by the end symbol 0 which is often omitted. Each instruction is prefixed with a \emph{line number} $l$ and a \emph{token} $t$. We call this pair an \emph{integrity key}. If $C_i$ is the body of a procedure in $\mathscr C$, then the token $t$ on every instruction in $C_i$ must be a \emph{token placeholder} $\Tok$. When the procedure is invoked, all token placeholders $\Tok$ in $C_i$ will be replaced with a fresh \emph{token value} $\tau$.

We assume that line numbers in $\mathscr C$ are similar to line numbers in a real computer program: Each instruction $I$ in $\mathscr C$ has a distinct line number $l$. When a procedure $X_i(\vx \p, \vx{\p.x}) = C_i$ is invoked, the line numbers in $C_i$ will remain unchanged. This will allow us to access the static location of an instruction at runtime in order to compute the integrity key.

There are five kinds of instructions. A communication $\BCom{\p}{e}{\q}{x};\,C$ instructs process $\p$ to evaluate expression $e$ and send it to process $\q$, which will bind the result to $\q.x$ in the continuation $C$. A selection $\BSelect{\p}{\q}{L};\,C$ conveys knowledge of choice~\cite{montesi2023}: it instructs $\p$ to send a value literal $L$ to $\q$, informing $\q$ that a decision (represented by $L$) has been made. A local computation $\BPure{x}{\p}{e};\,C$ instructs $\p$ to evaluate $e$ and bind the result to $\p.x$ in $C$. A conditional $\BIf{e}{\p}{C_1}{C_2};\,C$ instructs $\p$ to evaluate $e$ and for the processes to proceed with $C_1$ or $C_2$ according to the result. A procedure call $X_i(\vx \p, \vx a);\,C$ instructs processes $\vx \p$ to invoke procedure $X_i(\vx \q, \vx {\q.y}) = C_i$ defined in $\mathscr C$, with processes $\vx \p$ playing roles $\vx \q$ and arguments $\vx a$ (which may take the form of values $\At{v}{\p}$ or variables $\p.x$) substituted for parameters $\vx {\q.y}$ in $C_{i}$. In addition to these basic instructions, a choreography may contain blocks $\{\,C\,\};\,C'$ which limit the scope of variables defined in $C$ so they do not extend to $C'$.

In addition, choreographies can contain \emph{runtime instructions} that represent an instruction in progress; these terms are an artifact of the semantics, not written explicitly by the programmer. A communication-in-progress $\BRecv{\p}{\q}{x}$ indicates that $\p$ sent a message to $\q$, which $\q$ has not yet received. Similarly, a selection-in-progress $\BRecvSelect{\p}{\q}{L}$ indicates that $\p$ sent a selection. A procedure-call-in-progress $\BCalling{\vx \p}{X}{\vx \q, \vx a}{C}$ indicates that some processes have invoked $X$, and others have not---we leave the details to \Cref{sec:chor-semantics}.

Expressions $e$ are composed of \emph{atoms} $a$ (i.e.~variables $\p.x$ and values $\At{v}{p}$) and function applications $f(\vx e)$. Although the variables $\p.x$ are immutable, we assume that a function $f$ evaluated by $\p$ can mutate $\p$'s state as a side-effect. \changed{}{Technically, having side-effects in our theory is not necessary. However, most choreographic programming theories and implementations equip processes with mutable state~\cite{montesi2023}; this includes Choral, the language we use to implement the Ozone API in \Cref{sec:choral}.}

\newcommand{\BuyItem}{\val{BuyItem}}
\newcommand{\seller}{\proc{seller}}
\newcommand{\buyer}{\proc{buyer}}
\newcommand{\itemID}{\val{itemID}}
\newcommand{\itemm}{\val{item?}}
\newcommand{\sellItem}{\val{sell}}

\begin{figure}
\begin{subfigure}[t]{0.5\textwidth}
\begin{align*}
&\BuyItem(\proc s,\ \proc b,\ \proc b.\itemID) =\\
&\qquad \CCom{\proc b}{\itemID}{\proc s}{\itemID}{1, \Tok};\\
&\qquad \CPure{\itemm}{\proc s}{\sellItem(\proc s.\itemID)}{2, \Tok};\\
&\qquad \CCom{\proc s}{\itemm}{\proc b}{\itemm}{3, \Tok}\\
&\CCall{\BuyItem}{\seller, \buyer_1, \At{123}{\buyer_1}}{4,\tau_0};\\
&\CCall{\BuyItem}{\seller, \buyer_2, \At{543}{\buyer_2}}{5,\tau_0}
\end{align*}
\caption{}\label{fig:buy-item}
\end{subfigure}\hfill
\begin{subfigure}[t]{0.45\textwidth}
\begin{align*}
&\val{StreamIt}(\proc p,\ \proc c) =\\
&\qquad \CCom{\p}{produce()}{\proc c}{x}{1, \Tok};\\
&\qquad \CPure{z}{\proc c}{consume(\proc c.x)}{2, \Tok};\\
&\qquad \CIfM{(itemsLeft() > 0)}{\p}{3, \Tok} \CThen\\
&\qquad\qquad \CSelect{\proc p}{\proc c}{\small\textsc{More}}{4,\Tok}\\
&\qquad\qquad \CCall{StreamIt}{\proc p, \proc c}{5,\Tok}\\
&\qquad \CElse\\
&\qquad\qquad \CSelect{\proc p}{\proc c}{\small\textsc{Done}}{6,\Tok}\\
&\CCall{StreamIt}{\proc p_1, \proc c}{7,\tau_0};\\
&\CCall{StreamIt}{\proc p_2, \proc c}{8,\tau_0}
\end{align*}
\caption{}\label{fig:stream-it}
\end{subfigure}
    \caption{Two example choreographies. On the left, processes $\buyer_1$ and $\buyer_2$ concurrently attempt to buy products from $\seller$. On the right, producers $\p_1$ and $\p_2$ concurrently send streams of data to a shared consumer $\proc c$.}
    \label{fig:chor-examples}
\end{figure}

\subsection{Semantics}\label{sec:chor-semantics}

We now give a fully out-of-order semantics for choreographies in $O_3$. The semantics is a labelled transition system on \emph{configurations} $\CCfg{C}{\Sigma}{K}$, where $C$ is a choreography, $\Sigma$ is a mapping from process names $\p$ to process states $\sigma$, and $K$ is a mapping from process names $\p$ to multisets of messages $M$ yet to be delivered to $\p$. We also assume there exists a set of unchanging procedure declarations $\mathscr C$, not shown explicitly in the configuration.

An \emph{initial configuration} is a configuration $\CCfg{C}{\Sigma}{K}$ where $\Sigma$ maps each $\p$ to an arbitrary state,
$K$ maps each $\p$ to the empty set, and all instructions in $C$ use the same token $\tau_0$, called the \emph{initial token}. We assume initial configurations to be well-formed, cf.~\Cref{sec:chor-properties}. The transition relation $(\CStep{\p})$ is on configurations, where $\p$ identifies which process took a step.

Messages in our semantics are represented as triples $(l,\tau,v)$. Here $l$ is the line number of the communication that sent the message, $\tau$ is the token associated with the procedure invocation that sent the message, and $v$ is a value called the \emph{payload}. Together, the pair $(l,\tau)$ is called the \emph{integrity key} of the message; the line number prevents intraprocedural CIVs (\Cref{sec:intraprocedural-civs}) while the token prevents interprocedural CIVs (\Cref{sec:procedural-choreos}).

\subsubsection{Transition rules}

\Cref{fig:chor-semantics} defines the semantics for $O_3$, which extends textbook models for procedural and asynchronous choreographies to allow full out-of-order execution~\cite{montesi2023}. That is, in a choreography of the form $I_1;\,I_2;\,C$, the statement $I_2$ can always be executed before $I_1$ \emph{unless}:
\begin{enumerate}
    \item \emph{(Data dependency)} $I_1$ binds a variable $\p.x$ that is used in $I_2$; or 
    \item \emph{(Control dependency)} $I_1$ is a selection of the form $\BSelect{\p}{\q}{L}$ or $\BRecvSelect{\p}{\q}{L}$, and $I_2$ is an action performed by $\q$.
\end{enumerate}


The semantics for communication is defined by rules \textsc{C-Send} and \textsc{C-Recv}. In \textsc{C-Send} for the communication term $\CCom{\p}{e}{\q}{x}{l,\tau}$, the expression $e$ is evaluated in the context of $\p$'s state using the notation $\CEval{\Sigma(\p)}{e}{(v,\sigma)}$. Evaluating $e$ produces a value $v$ and a new state $\sigma$ for $\p$; we assume that $(\vdash)$ is defined for any $e$ that contains no free variables and for any state $\Sigma(\p)$. The \textsc{C-Send} rule transforms the communication term into a communication-in-progress term $\CRecv{\p}{\q}{x}{l,\tau}$ and adds the message $(l,\tau,v)$ to $\q$'s set of undelivered messages. The message can subsequently be received by $\q$ using the \textsc{C-Recv} rule. This rule removes the communication-in-progress term and substitutes the message payload $v$ into the continuation $C$. Notice that the integrity key $l,\tau$ of the message is matched against the integrity key of the communication-in-progress, $\CRecv{\p}{\q}{x}{l,\tau}$. Notice also that the semantics for communication is not defined if the token $t$ is merely a placeholder $\Tok$---it must be a token \emph{value} $\tau$. Indeed, in \Cref{sec:chor-properties} we show that placeholders only appear in $\mathscr C$, never in $C$.

Rules \textsc{C-Select} and \textsc{C-OnSelect} closely mirror the semantics of \textsc{C-Send} and \textsc{C-Recv}---the key difference is that a label $L$ is communicated instead of a value. Rules \textsc{C-Compute} and \textsc{C-If} are standard, except for changes made to use lexical scope instead of global scope: \textsc{C-Compute} substitutes the value $v$ into the continuation $C$ (instead of storing it in the local state $\Sigma$) and \textsc{C-If} places the continuation $C_i$ in a block to prevent variable capture. To garbage collect empty blocks, \textsc{C-If} uses a concatenation operator $(\fatsemi)$ defined as: 
$$\{I;C\}\fatsemi C' = \{I;C\};\,C' \qquad\qquad \{0\}\fatsemi C' = C'$$

The \textsc{C-Delay} rule is used in choreography models to enable a limited form of out-of-order execution, where unrelated processes execute concurrently: given a choreography $I;\,C$, \textsc{C-Delay} would ordinarily prevent any $\q$ from executing in $C$ if $\q$ is somehow involved in $I$. Our formulation of the rule is weakened: $\q$ is only prevented from executing in $C$ if $I$ is a selection at $\q$, i.e.~a control dependency. The rule still respects data dependencies, however, by design of the other rules---for instance, $\CCom{\p}{x}{\q}{y}{l,\tau}$ cannot be evaluated until $x$ is bound to a value. Thus our version of \textsc{C-Delay} enables \emph{full} out-of-order execution.

The rules \textsc{C-First}, \textsc{C-Enter}, \textsc{C-Last}, and \textsc{C-Delay-Proc} model procedure calls, with extra machinery to model how processes can execute their roles in a choreographic procedure in parallel until they need to interact. Given a procedure call $\CCall{X}{\vx\p,\vx a}{l,\tau}$, \textsc{C-First} models how $\p \in \vx \p$ has entered the procedure before any of the other processes. The rule replaces the procedure call with a procedure-call-in-progress $\CCalling{\vx\p\setminus \p}{X}{\vx \p,\vx a}{C_1'}{l,\tau}$ to reflect this fact; the choreography $C_1'$ is the body of the procedure, which $\p$ may begin executing via the \textsc{C-Delay-Proc} rule. The remaining processes can enter the procedure via the \textsc{C-Enter} rule, and the last process to enter the procedure uses the \textsc{C-Last} rule. As we explain below, these rules also compute new integrity keys for the callee procedure to prevent CIVs.

The key novelty of our semantics for procedures is the use of $\NextToken$. In \textsc{C-First}, the body $C_{1}'$ is obtained by computing the token $\tau' = \NextToken(l,\tau)$ and substituting $\tau'$ for all occurrences of the token placeholder $\Tok$. Notice that the semantics makes it appear as if the processes have synchronized to compute the next token; in \Cref{sec:processes}, we give a semantics where each process computes the next token independently and in \Cref{thm:epp} we prove that the two models correspond. Hence the apparent synchronization has no runtime cost.

As discussed in \Cref{sec:interprocedural-civs}, $\NextToken: \mathbb N \times \TokenType \to \TokenType$ is a pure injective function for computing new tokens (of type $\TokenType$) using integrity keys (of type $\mathbb N \times \TokenType$). To ensure the integrity keys from two concurrent procedures never collide, $\NextToken$ must produce unique, non-repeating keys upon iterated application. One way this can be realized is by representing $\TokenType = \mathbb N^*$ as lists of numbers, the initial token $\tau_0$ as an empty list $[]$, and implementing $\NextToken(l, \tau) = l :: \tau$, i.e.\ prepending the line number $l$ to the list. Intuitively, this means the token associated with a procedure invocation is a simplified \emph{call stack} of line numbers from which the procedure was called.

\begin{figure}
    \centering\small
    
\begin{prooftree}
    \AxiomC{$\CEval{\Sigma(\p)}{e}{(v,\sigma)}$}
    \AxiomC{$M = K(\q) \uplus \{(l,\tau,v)\}$}
    \RightLabel{\textsc{C-Send}}
    \BinaryInfC{ $\CCfg{ \CCom{\p}{e}{\q}{x}{\Line,\tau};\,C }{\Sigma}{K} \CStep{\p} \CCfg{ \CRecv{\p}{\q}{x}{\Line,\tau};\,C }{\Sigma[\p \mapsto \sigma]}{K[\q \mapsto M]}$ }
\end{prooftree}

\begin{prooftree}
    \AxiomC{$(\Line,\tau, v) \in K(\q)$}
    \AxiomC{$M = K(\q) \setminus \{(\Line,\tau, v)\}$}
    \RightLabel{\textsc{C-Recv}}
    \BinaryInfC{ $\CCfg{ \CRecv{\p}{\q}{x}{\Line,\tau};\,C }{\Sigma}{K} \CStep{\q} \CCfg{ C[\q.x \mapsto \At{v}{q}] }{\Sigma}{K[\q \mapsto M]}$ }
\end{prooftree}

\begin{prooftree}
    \AxiomC{$M = K(\q) \cup \{(l,\tau,L)\}$}
    \RightLabel{\textsc{C-Select}}
    \UnaryInfC{ $\CCfg{ \CSelect{\p}{\q}{L}{\Line,\tau};\,C }{\Sigma}{K} \CStep{\p} \CCfg{ \CRecvSelect{\p}{\q}{L}{\Line,\tau};\,C }{\Sigma}{K[\q \mapsto M]}$ }
\end{prooftree}

\begin{prooftree}
    \AxiomC{$K(\q) = \{(\Line,\tau, L)\} \cup M$}
    \RightLabel{\textsc{C-OnSelect}}
    \UnaryInfC{ $\CCfg{ \CRecvSelect{\p}{\q}{L}{\Line,\tau};\,C }{\Sigma}{K} \CStep{\q} \CCfg{ C }{\Sigma}{K[\q \mapsto M]}$ }
\end{prooftree}

\begin{prooftree}
    \AxiomC{$\CEval{\Sigma(\p)}{e}{(v,\sigma)}$}
    \RightLabel{\textsc{C-Compute}}
    \UnaryInfC{ $\CCfg{ \CPure{x}{\p}{e}{\Line,\tau};\,C }{\Sigma}{K} \CStep{\p} \CCfg{ C[\p.x \mapsto \At{v}{\p}] }{\Sigma[\p \mapsto \sigma]}{K}$ }
\end{prooftree}

\begin{prooftree}
    \AxiomC{$\CEval{\Sigma(\p)}{e}{v}$}
    \AxiomC{if $v = \True$ then $i=1$ else $i=2$}
    \RightLabel{\textsc{C-If}}
    \BinaryInfC{ $\CCfg{ \CIf{e}{\p}{C_1}{C_2}{\Line,\tau};\,C }{\Sigma}{K} \CStep{\p} \CCfg{ \Block{C_i}\fatsemi\,C }{\Sigma}{K}$ }
\end{prooftree}

\begin{prooftree}
    \AxiomC{ $\CCfg{ C_1 }{\Sigma}{K} \CStep{\p} \CCfg{ C_1' }{\Sigma'}{K'}$ }
    \RightLabel{\textsc{C-Block}}
    \UnaryInfC{ $\CCfg{ \Block{C_1};\,C_2 }{\Sigma}{K} \CStep{\p} \CCfg{ \Block{C_1'}\fatsemi\,C_2 }{\Sigma'}{K'}$ }
\end{prooftree}

\begin{prooftree}
    \AxiomC{ $\CCfg{ C }{\Sigma}{K} \CStep{\q} \CCfg{ C' }{\Sigma'}{K'}$ }
    \AxiomC{ $I$ is not a selection at $\q$ }
    \RightLabel{\textsc{C-Delay}}
    \BinaryInfC{ $\CCfg{ I;\,C }{\Sigma}{K} \CStep{\q} \CCfg{ I;\,C' }{\Sigma'}{K'}$ }
\end{prooftree}

\begin{prooftree}
    \noLine
    \AxiomC{ $(X(\vx \q,\vx{\q.y}) = C_1) \in \mathscr C$ \qquad $C_1' = C_1[\vx \q, \vx{\q.y}, \Tok \mapsto \vx \p, \vx a, \tau']$  }
    \UnaryInfC{$\p \in \vx \p$ \qquad\qquad\qquad\qquad $\tau' = \NextToken(\Line,\tau)$}
    \RightLabel{\textsc{C-First}}
    \UnaryInfC{ $\CCfg{ \CCall{X}{\vx \p, \vx a}{\Line,\tau};\,C_2 }{\Sigma}{K} \CStep{\p} \CCfg{ \CCalling{\vx \p \setminus \p}{X}{\vx \p, \vx a}{C_1'}{\Line,\tau};\,C_2 }{\Sigma}{K}$ }
\end{prooftree}

\begin{prooftree}
    \AxiomC{ $\p \in \vx \p$ }
    \RightLabel{\textsc{C-Enter}}
    \UnaryInfC{ $\CCfg{ \CCalling{\vx \p}{X}{\vx \q, \vx a}{C_1}{\Line,\tau};\,C_2 }{\Sigma}{K} \CStep{\p} \CCfg{ \CCalling{\vx \p \setminus \p}{X}{\vx \q, \vx a}{C_1}{\Line,\tau};\,C_2 }{\Sigma}{K}$ }
\end{prooftree}

\begin{prooftree}
    \AxiomC{  }
    \RightLabel{\textsc{C-Last}}
    \UnaryInfC{ $\CCfg{ \CCalling{\p}{X}{\vx \q, \vx a}{C_1}{\Line,\tau};\,C_2 }{\Sigma}{K} \CStep{\p} \CCfg{ \Block{C_1}\fatsemi\,C_2 }{\Sigma}{K}$ }
\end{prooftree}

\begin{prooftree}
    \AxiomC{ $\CCfg{ C_1 }{\Sigma}{K} \CStep{\p} \CCfg{ C_1' }{\Sigma'}{K'}$ }
    \AxiomC{ $\p \notin \vx \p$}
    \RightLabel{\textsc{C-Delay-Proc}}
    \BinaryInfC{ $\CCfg{ \CCalling{\vx\p}{X}{\vx \q, \vx a}{C_1}{\Line,\tau};\,C_2 }{\Sigma}{K} \CStep{\p} \CCfg{ \CCalling{\vx\p}{X}{\vx \q, \vx a}{C_1'}{\Line,\tau};\,C_2 }{\Sigma'}{K'}$ }
\end{prooftree}
    \caption{Semantics for fully out-of-order choreographies}%
    \label{fig:chor-semantics}
\end{figure}

\subsubsection{Discussion}\label{sec:chor-discussion}

\Cref{fig:buy-item} expresses a choreography in which two $\buyer$ processes concurrently buy items from a $\seller$ process. In the initial configuration, $\buyer_1$ can enter the procedure on line 4, $\buyer_2$ can enter the procedure on line 5, and $\seller$ can enter either procedure. If $\buyer_2$ enters first (using \textsc{C-Delay} and \textsc{C-Enter}), it can proceed to send $\At{543}{\buyer_2}$ to $\seller$ (using \textsc{C-Com}). Then $\seller$ can enter the procedure on line 5 (using \textsc{C-Delay} and \textsc{C-Last}) and proceed to receive the message from $\buyer_2$ (using \textsc{C-Recv}). This execution would be impossible in a standard choreography model because $\seller$ would need to complete the procedure invocation on line 4 before it could enter the procedure on line 5. The added concurrency ensures that slowness in $\buyer_1$ does not prevent $\buyer_2$ from making progress.

Notice the out-of-order semantics of \Cref{fig:buy-item} also adds nondeterminism. Suppose $\buyer_1$ and $\buyer_2$ attempt to buy the same item and the $\seller$ only has one copy. One of the buyers will receive the item, and the other will receive a null value. In a standard choreography model, the item would always go to $\buyer_1$. In $O_3$, the item will be sold nondeterministically according to the order that messages arrive to the seller. This nondeterminism can be problematic---it makes reasoning about choreographies harder---but also increases expressivity: nondeterminism is essential in distributed algorithms like consensus and leader election. Reasoning about nondeterminism in choreographies is an important topic for future work.

\Cref{fig:stream-it} shows we can also express recursive choreographies. In each iteration of the procedure \emph{StreamIt}, a producer $\proc p$ sends a value to a consumer $\proc c$ (line 1) and decides whether to start another iteration (line 3).  Then the producer asynchronously informs the consumer about its decision (lines 4 and 6) and can proceed with the next iteration (line 5) without waiting for the consumer. Because messages in $O_3$ are unordered, the consumer can consume items (line 2) from different iterations in any order; this prevents head-of-line blocking~\cite{scharf2006}.

In the initial choreography of \Cref{fig:stream-it}, producers $\p_1,\p_2$ and a consumer $\proc c$ invoke two instances of \emph{StreamIt}. As in \Cref{fig:buy-item}, the two procedures evolve concurrently; a slowdown in $\p_1$ will not prevent $\proc c$ from consuming items produced by $\p_2$.

\subsection{Properties}\label{sec:chor-properties}

In this section we prove that $O_3$ choreographies are deadlock-free and we formalize the communication integrity property.  Combined with the EPP Theorem presented in \Cref{sec:processes}, these results imply that projected code inherits the same properties.

\begin{figure}[ht]
  \centering\small

\begin{minipage}{0.5\textwidth}
\begin{prooftree}
    \AxiomC{$\forall v,\,(l,\tau,v) \notin K(\q)$}
    \RightLabel{\textsc{C-WF-Send}}
    \UnaryInfC{ $\langle \CCom{\p}{e}{\q}{x}{\Line,\tau},\, {K} \rangle \wf$ }
\end{prooftree}
\end{minipage}%
\begin{minipage}{0.5\textwidth}
\begin{prooftree}
    \AxiomC{$(l,\tau,L) \notin K(\q)$}
    \RightLabel{\textsc{C-WF-Select}}
    \UnaryInfC{ $\langle \CSelect{\p}{\q}{L}{\Line,\tau},\, {K} \rangle \wf$ }
\end{prooftree}
\end{minipage}%

\begin{prooftree}
    \noLine
    \AxiomC{$\vx\p$ distinct \qquad $\vx{\p.x}$ distinct \qquad $\pn(C) \subseteq \vx\p$}
    \UnaryInfC{$\forall \p.x \in \vx{\p.x},\,\p \in \vx\p$ \qquad $\langle I, K\rangle \wf$ for each $I \in \stats(C)$}
    \noLine
    \UnaryInfC{$C$ contains no runtime terms \qquad $\keys(C)$ distinct \qquad $\forall (l,t) \in \keys(C),\,t = \Tok$}
    \RightLabel{\textsc{C-WF-Def}}
    \UnaryInfC{ $X(\vx\p, \vx{\p.x}) = C \wf$ }
\end{prooftree}

\begin{prooftree}
    \noLine
    \AxiomC{ $\langle \CCall{X}{\vx \q, \vx a}{\Line,\tau},\, {K} \rangle \wf$ \qquad $(X(\q_1,\dots,\q_n, \q^1.x_1, \dots, q^m.x_m) = C') \in \mathscr C$ }
    \UnaryInfC{ $\{\r_1,\dots,\r_k\} \subseteq \{\p_1,\dots,\p_n\}$ \qquad $\forall i \le k, j \le n\,$ if $\r_i = \p_j$ then $\EPP{C}{\r_i} = \EPP{C'}{\q_j}$}
    \RightLabel{\textsc{C-WF-Calling}}
    \UnaryInfC{ $\langle \CCalling{\r_1,\dots,\r_k}{X}{\p_1,\dots,\p_n, a_1,\dots,a_m}{C}{\Line,\tau},\, {K} \rangle \wf$ }
  \end{prooftree}
\begin{minipage}{0.6\textwidth}
    \begin{align*}
        &\stats(0) = \epsilon\\
        &\stats(I;\,C) = \stats(I),\, \stats(C)\\
        &\stats(\{\,C\,\}) = \stats(C)\\
        &\stats(\CIf{e}{\q}{C_1}{C_2}{l,t}) =\\
            &\quad(\CIf{e}{\q}{C_1}{C_2}{l,t}),\, \stats(C_1),\, \stats(C_2)\\
        &\stats(\CCalling{\vx\q}{X}{\vx\p, \vx a}{C}{l,t}) =\\
      &\quad(\CCalling{\vx\q}{X}{\vx\p, \vx a}{C}{l,t}), \stats(C)\\
        &\stats(l,t : \eta) = (l,t : \eta) \ \hbox{otherwise}\\
        &\stats(C) = [\stats(C)\ |\ \p \in \pn(C)]\\
        &\keys(C) = [ (l,t)\ |\ (l,t:\eta) \in \stats(C)]
    \end{align*}
  \end{minipage}%
  \begin{minipage}{0.40\textwidth}
    \begin{align*}
    &\pn(0) = \emptyset\\
    &\pn(I;\,C) = \pn(I) \cup \pn(C)\\
    &\pn(\{\,C\,\}) = \pn(C)\\
    &\pn(\CCom{\p}{e}{\q}{x}{l,t}) = \{\p,\q\}\\
    &\pn(\CRecv{\p}{\q}{x}{l,t}) = \{\q\}\\
    &\pn(\CSelect{\p}{\q}{L}{l,t}) = \{\p,\q\}\\
    &\pn(\CRecvSelect{\p}{\q}{L}{l,t}) = \{\q\}\\
    &\pn(\CPure{x}{\p}{e}{l,t}) = \{\p\}\\
    &\pn(\CIf{e}{\p}{C_1}{C_2}{l,t}) =\\
    &\quad\{\p\} \cup \pn(C_1) \cup \pn(C_2)\\
    &\pn(\CCall{X}{\vx\p,\vx a}{l,t}) = \vx\p\\
    &\pn(\CCalling{\vx\q}{X}{\vx\p, \vx a}{C}{l,t}) = \vx \p\\
    &\pn(\At{v}{\p}) = \{\p\}\\
    &\pn(\p.x) = \{\p\}
    \end{align*}
  \end{minipage}%

      \caption{Well-formedness (representative rules)}%
  \label{fig:well-formedness-main-rules}
\end{figure}

To prove these properties we need an invariant that characterizes how the rules of $O_{3}$ preserve the intuition from \Cref{sec:motivation}. For example, consider the following configurations:
\begin{gather}
\CCfg{\CRecv{\p}{\q}{x}{l,\tau_0}}{\Sigma}{\{\p \mapsto \emptyset, \q \mapsto \emptyset\}}
\label{eq:conf-1}
\\
\CCfg{\CCom{\p}{e}{\q}{x}{l,\tau_0}}{\Sigma}{\{\p \mapsto \emptyset, \q \mapsto \{(l,\tau,v)\}\}}
\label{eq:conf-2}
\\
\CCfg{\{\CCom{\p}{e}{\q}{x}{1,\tau} \};\,\{\CCom{\p}{e'}{\q}{x}{1,\tau} \}}{\Sigma}{\{\p \mapsto \emptyset, \q \mapsto \emptyset\}}
\label{eq:conf-3}
\\
\CCfg{\CCalling{\p}{X}{\p,\q}{ \CCom{\p}{e}{\q}{x}{1,\tau_0} }{3,\tau_0}}{\Sigma}{\{\p \mapsto \emptyset, \q \mapsto \emptyset\}}
\label{eq:conf-4}
\end{gather}
Configuration~\eqref{eq:conf-1} is not reachable because $\CRecv{\p}{\q}{x}{l,\tau}$ never occurs unless $\q$ has an undelivered message from $\p$. Dually, configuration~\eqref{eq:conf-2} is not reachable because $\p$ has a message in its queue that, according to the choreography, has not yet been sent. Configuration~\eqref{eq:conf-3} is unreachable because the two instructions share the same integrity key; we will show that $\NextToken$ ensures such configurations never arise. Likewise, $\NextToken$ also forbids configuration~\eqref{eq:conf-4}, since the token of the instruction $\CCom{\p}{e}{\q}{x}{1,\tau_0}$ must have been derived from the integrity key of the enclosing call $\CCalling{\p}{X}{\p,\q}{ \dots }{3,\tau_0}$. Specifically, $\tau_0 \ne \NextToken(3,\tau_0)$. To specify this last property, recall that tokens are represented as lists of integers $l_1 :: l_2 :: \dots$. We say $(l_1,t_1)$ is a \emph{prefix} of $(l_2,t_2)$---written $(l_1,t_1) \prec (l_2,t_2)$---if the list $l_1::t_1$ is a prefix of $l_2::t_2$ and that the keys are \emph{disjoint} if neither is a prefix of the other.

Following convention, we formalize the properties of reachable configurations by defining which configurations and procedures are \emph{well-formed}.
\Cref{fig:well-formedness-main-rules} highlights the most interesting rules that define well-formedness, where $\wf$ reads `well-formed' -- the rest can be found in \Cref{sec:well-formedness}.
In particular, well-formedness ensures that:
\begin{enumerate}
    \item (\textsc{C-WF-Send}) $\CRecv{\p}{\q}{x}{l,\tau}$ occurs in $C$ if and only if $(l,\tau,v) \in K(\q)$ for some $v$.
    \item (\textsc{C-WF-Select}) $\CRecvSelect{\p}{\q}{L}{l,\tau}$ occurs in $C$ if and only if $(l,\tau,L) \in K(\q)$.
    \item (\textsc{C-WF-Def}) Each $I$ in $C$ has a distinct integrity key $l,t$, where $t$ is not a placeholder.
    \item (\textsc{C-WF-Calling}) If the integrity key of $I$ is a prefix of the integrity key of $I'$ then $I$ is a communication-in-progress $\CCalling{\vx\p}{X}{\vx \p, \vx a}{C'}{l,t}$ and $I'$ is in $C'$.
\end{enumerate}
Well-formedness also guarantees other properties seen in other choreography models, e.g., that procedures contain no free variables and that processes waiting to enter a procedure have the same local behaviour in the original procedure body and the current choreography~\cite{montesi2023}. As in prior work~\cite{montesi2023,CMP23}, the latter check is made by using endpoint projection ($\EPP{C}{\p}$), which returns the local behaviour of a process in a choreography and is defined later in~\Cref{sec:epp}.

\begin{theorem}[Preservation]\label{thm:preservation}
    If $\CCfg{C}{\Sigma}{K}$ is well-formed and $\CCfg{C}{\Sigma}{K} \CStep{\p} \CCfg{C'}{\Sigma'}{K'}$, then $\CCfg{C'}{\Sigma'}{K'}$ is well-formed.
\end{theorem}
\begin{proof}
    By induction on the rules of $\CStep{\p}$. We focus on the rules for communication and procedure invocation.

    \textsc{C-Send} replaces $\CCom{\p}{e}{\q}{x}{l,\tau}$ with $\CRecv{\p}{\q}{x}{l,\tau}$ and adds a message $(l,\tau,v)$. By the induction hypothesis, $(l,\tau,v)$ is not already in $K$.

    \textsc{C-Recv} eliminates $\CRecv{\p}{\q}{x}{l,\tau}$ and removes a message $(l,\tau,v)$. Since each instruction has a distinct integrity key by hypothesis, no other $\CRecv{\p}{\q}{x}{l,\tau}$ term occurs in $C$.

    \textsc{C-First} introduces new terms into the choreography by invoking the call $\CCall{X}{\vx\p, \vx a}{l_1,\tau_1}$. By the induction hypothesis, for any other instruction $l_2,t_2:I$ in $C$, either (a) keys $l_1,t_1$ and $l_2,t_2$ are disjoint; or (b) $l_2,t_2: I$ is a call-in-progress containing $\CCall{X}{\vx\p, \vx a}{l_1,\tau_1}$. In case (a), disjointness implies any instruction in the body of the procedure $C'[\vx{\q}, \vx{\q.y}, \Tok \mapsto \vx{\p}, \vx{\p.x}, \tau']$ will also have a key that is disjoint from $l_2,t_2$. In case (b), notice $\forall l,\, (l_2,t_2) \prec (l_1,t_1) \prec (l, \NextToken(l_1,t_1))$; hence any interaction in the body has a key where $(l_2,t_2)$ is a prefix.
\end{proof}

\begin{theorem}[Deadlock-Freedom]\label{thm:progress}
    If $\CCfg{C}{\Sigma}{K}$ is well-formed, then either $C \equiv 0$ or $\CCfg{C}{\Sigma}{K} \CStep{\p} \CCfg{C'}{\Sigma'}{K'}$ for some $\p,C',\Sigma',K'$.
\end{theorem}
\begin{proof}
    By induction on the structure of $C$, making use of the full definition of well-formedness in \Cref{sec:well-formedness}. In each case, we observe the first instruction $I$ of $C$ can always be executed. For instance, if $I \equiv \CCom{\p}{e}{\q}{x}{l,\tau}$ then the \textsc{C-Send} rule can be applied because well-formedness implies $e$ has no free variables. If $I \equiv \CRecv{\p}{\q}{x}{l,\tau}$, there must be a message $(l,\tau,v) \in K(\q)$ because the configuration is well-formed. The other cases follow similarly.
\end{proof}

We end this section with a formalization of communication integrity. Consider the buggy execution in \Cref{fig:forwarding-data}: in a model without integrity keys, the execution reaches a configuration 
$$\CCfg{\BRecv{\proc s}{\proc c}{txt};\,\BRecv{\proc s}{\proc c}{key};\,\dots}{\Sigma}{\proc c \mapsto v_{key}, v_{txt}},$$
where $v_{key}$ is the value produced by $\proc{ks}.getKey()$ and $v_{txt}$ is the value produced by $\proc{cs}.getText()$. A CIV occurs if the configuration can make a transition that consumes $\BRecv{\proc s}{\proc c}{txt}$ and $v_{key}$ together, binding $\proc{c}.txt$ to $v_{key}$. We therefore want to ensure:
\begin{itemize}
    \item There is only one way a communication-in-progress instruction can be consumed; and
    \item The instruction is consumed together with the correct message.
\end{itemize}

\begin{definition}[Send/receive transitions]
    A \emph{send transition} \(\CCfg{C}{\Sigma}{K} \CStep{\p} \CCfg{C'}{\Sigma'}{K'}\) is a transition with a derivation that ends with an application of \emph{\textsc{C-Send}}. Likewise, a \emph{receive transition} is a transition with a derivation that ends with \emph{\textsc{C-Recv}}.
\end{definition}

\begin{theorem}[Communication Integrity]\label{thm:integrity}
    Let $e = c_0 \CStep{\p_1} \cdots \CStep{\p_{k+1}} c_{k+1}$ be an execution ending with a send transition $c_{k} \CStep{\p} c_{k+1}$, which produces instruction $\CRecv{\p}{\q}{x}{l,\tau}$ and message $m$. Let $e' = c_0 \CStep{\p_1} \cdots \CStep{\p_n} c_n\ (n > k)$ be an execution extending $e$, where $\CRecv{\p}{\q}{x}{l,\tau}$ has not yet been consumed. Then there is at most one receive transition $c_n \CStep{\q} c_{n+1}$ consuming $\CRecv{\p}{\q}{x}{l,\tau}$. Namely, it is the transition that consumes $\CRecv{\p}{\q}{x}{l,\tau}$ and $m$ together.
\end{theorem}
\begin{proof}
    By definition of \textsc{C-Send}, $m$ has the form $(l,\tau,v)$. By definition of \textsc{C-Recv}, if there exists a transition $c_n \to c_{n+1}$ that consumes $\CRecv{\p}{\q}{x}{l,\tau}$, then the transition also consumes a message $(l,\tau,v')$, for some $v'$. It therefore suffices to show the message $(l,\tau,v')$ is unique and that $v' = v$. This follows by induction on the length $m$ of the extension:
    \begin{itemize}
        \item \emph{Base case:} Well-formedness implies there is no message $(l,\tau,v')$ in $c_k$. Hence the message $(l,\tau,v)$ in $c_{k+1}$ is unique.
        \item \emph{Induction step:} Observe that the transition $c_m \to c_{m+1}$ cannot remove $(l,\tau,v)$; this would require consuming $\CRecv{\p}{\q}{x}{l,\tau}$, which cannot happen in $e'$ by hypothesis. Also observe that the transition cannot add a new message with integrity key $(l,\tau)$; this would require consuming an instruction $\CCom{\p'}{e}{\q}{x'}{l,\tau}$, which cannot exist in $c_m$ by well-formedness. Hence $(l,\tau,v)$ is unique in $c_{m+1}$.
    \end{itemize}
\end{proof}

\section{Process Model and Endpoint Projection}\label{sec:processes}

\subsection{Syntax}\label{sec:proc-syntax}

\Cref{fig:proc-syntax} presents the syntax for out-of-order processes. A term $\p[P]$ is a process named $\p$ with behavior $P$. Networks, ranged over by $N,M$, are parallel compositions of processes. Compared to prior work, certain process instructions need to be annotated with integrity keys (for instance, message send $\PSend{\p}{l}{t}{e}$ and procedure call $\PCall{X}{\vx \p, \vx a}{l}{t}$). In addition, when receiving a message it is no longer necessary to specify a sender---it suffices to write $\PRecv{x}{l}{t};\,P$ instead of the more traditional $\p\,\PRecv{x}{l}{t};\,P$---because integrity keys functionally determine the variable to which the message payload should be bound.

\begin{figure}
    \centering\small
\begin{subfigure}{0.90\textwidth}
\begin{align*}
\mathscr P &::= \{X_{i}(\vx \p_{i}, \vx x_{i}) = C_{i}\}_{i \in \mathcal I}&\emph{(decls)}&&\\
P, Q &::= I;\, P &\emph{(seq)}&\quad\alt \{\,P\,\} & \emph{(block)}\\
&\alt 0 & \emph{(end)}&&\\
I &::= \PSend{\p}{\Line}{t}{e} & \emph{(send)}&\quad\alt \PRecv{x}{\Line}{t} &\emph{(receive)}\\
&\alt \PSet{x}{e} &\emph{(expr)}&\quad\alt \PSelect{\p}{\Line}{t}{L} &\emph{(choice)}\\
&\alt \PBranch{l_i,\tau,L_i}{P_i}{i \in \mathcal I} &\emph{(branch)}&\quad\alt \PIf{e}{P}{Q} &\emph{(cond)}\\
&\alt \PCall{X}{\vx \p, \vx a}{\Line}{t} & \emph{(call)}&&\\
e &::= f(\vx e) &\emph{(app)}&\quad\alt a &\emph{(atom)}\\
a &::= x &\emph{(var)}&\quad\alt v&\emph{(val)}\\
N,M &::= \p[P] &\emph{(proc)} &\quad\alt (N\ |\ M) &\emph{(par)}
\end{align*}   
\end{subfigure}
    \caption{Syntax for out-of-order processes}
    \label{fig:proc-syntax}
\end{figure}

\subsection{Semantics}\label{sec:proc-semantics}

\begin{figure}
    
\begin{prooftree}
    \AxiomC{$\CEval{\Sigma(\p)}{e}{(v,\sigma)}$}
    \AxiomC{$M = K(\q) \uplus \{(l,\tau,v)\}$}
    \RightLabel{\textsc{P-Send}}
    \BinaryInfC{ $\PCfg{ \p[\PSend{\q}{\Line}{\tau}{e};\,P] }{\Sigma}{K} \PStep{\p} \PCfg{ \p[P] }{\Sigma[\p \mapsto \sigma]}{K[\q \mapsto M]}$ }
\end{prooftree}

\begin{prooftree}
    \AxiomC{$(\Line,\tau, v) \in K(\q)$}
    \AxiomC{$M = K(\q) \setminus \{(\Line,\tau, v)\}$}
    \RightLabel{\textsc{P-Recv}}
    \BinaryInfC{ $\PCfg{ \q[\PRecv{x}{\Line}{\tau};\,Q] }{\Sigma}{K} \PStep{\q} \PCfg{ \q[Q[x \mapsto v]] }{\Sigma}{K[\q \mapsto M]}$ }
\end{prooftree}

\begin{prooftree}
    \AxiomC{$M = K(\q) \cup \{(l,\tau,L)\}$}
    \RightLabel{\textsc{P-Select}}
    \UnaryInfC{ $\PCfg{ \p[\PSelect{\q}{\Line}{\tau}{L};\,P] }{\Sigma}{K} \PStep{\p} \PCfg{ \p[P] }{\Sigma}{K[\q \mapsto M]}$ }
\end{prooftree}

\begin{prooftree}
    \AxiomC{$K(\q) = \{(\Line_i,\tau,L_i)\} \cup M$}
    \AxiomC{$i \in \mathcal I$}
    \RightLabel{\textsc{P-OnSelect}}
    \BinaryInfC{ $\PCfg{ \q[\PBranch{l_j,\tau,L_j}{Q_j}{j \in \mathcal I};\,Q] }{\Sigma}{K} \PStep{\q} \PCfg{ \q[\{Q_i\};\,Q] }{\Sigma}{K[\q \mapsto M]}$ }
\end{prooftree}

\begin{prooftree}
    \AxiomC{$\CEval{\Sigma(\p)}{e}{(v,\sigma)}$}
    \RightLabel{\textsc{P-Compute}}
    \UnaryInfC{ $\PCfg{ \p[\PSet{x}{e};\,P] }{\Sigma}{K} \PStep{\p} \PCfg{ \p[P[x \mapsto v]] }{\Sigma[\p \mapsto \sigma]}{K}$ }
\end{prooftree}

\begin{prooftree}
    \AxiomC{$\CEval{\Sigma(\p)}{e}{v}$}
    \AxiomC{if $v = \True$ then $i = 1$ else $i = 2$}
    \RightLabel{\textsc{P-If}}
    \BinaryInfC{ $\PCfg{ \p[\PIf{e}{P_1}{P_2};\,P] }{\Sigma}{K} \PStep{\p} \PCfg{ \p[\{P_i\};\,P] }{\Sigma}{K}$ }
\end{prooftree}

\begin{prooftree}
    \AxiomC{ $\PCfg{ \p[P_1] }{\Sigma}{K} \PStep{\p} \PCfg{ \p[P_1'] }{\Sigma'}{K'}$ }
    \RightLabel{\textsc{P-Block}}
    \UnaryInfC{ $\PCfg{ \p[\{P_1\};\,P_2] }{\Sigma}{K} \PStep{\p} \PCfg{ \p[\{P_1'\};\,P_2] }{\Sigma'}{K'}$ }
\end{prooftree}

\begin{prooftree}
    \AxiomC{ $\PCfg{ \p[P] }{\Sigma}{K} \PStep{\p} \PCfg{ \p[P'] }{\Sigma'}{K'}$ }
    \RightLabel{\textsc{P-Delay}}
    \UnaryInfC{ $\PCfg{ \p[I;\,P] }{\Sigma}{K} \PStep{\p} \PCfg{ \p[I;\,P'] }{\Sigma'}{K'}$ }
\end{prooftree}

\begin{prooftree}
    \AxiomC{ $(X(\vx \q,\vx y) = Q) \in \mathscr P$ }
    \AxiomC{ $\NextToken(l,\tau) = \tau'$ }
    \RightLabel{\textsc{P-Call}}
    \BinaryInfC{ $\PCfg{ \p[\PCall{X}{\vx\p,\vx a}{l}{\tau};\,P] }{\Sigma}{K} \PStep{\p} \PCfg{ \p[\{Q[\vx \q, \vx y, \Tok \mapsto \vx \p, \vx a, \tau']\};\,P] }{\Sigma}{K}$ }
\end{prooftree}

\begin{prooftree}
    \AxiomC{ $\PCfg{ N }{\Sigma}{K} \PStep{\p} \PCfg{ N' }{\Sigma'}{K'}$ }
    \RightLabel{\textsc{P-Par}}
    \UnaryInfC{ $\PCfg{ N \Par M }{\Sigma}{K} \PStep{\p} \PCfg{ N' \Par M }{\Sigma'}{K'}$ }
  \end{prooftree}
  \caption{Semantics for out-of-order processes}
    \label{fig:proc-semantics}
\end{figure}

The semantics for out-of-order processes appears in \Cref{fig:proc-semantics}. It is a labelled transition system on \emph{process configurations} $\PCfg{N}{\Sigma}{K}$, where $N$ is a network and $\Sigma,K$ have the same meaning as in \Cref{sec:chor-semantics}. We also assume an implicit set of procedure declarations $\mathscr P$.

The transition rules of \Cref{fig:proc-semantics} are similar to prior work. \textsc{P-Send} adds a message $(l,\tau,v)$ to the undelivered messages of $\q$, whereas \textsc{P-Recv} removes the message and substitutes it into the body of the process. Similarly, \textsc{P-Select} adds $(l,\tau,L)$ to the message set and \textsc{P-OnSelect} selects a branch from the set of options $\PBranch{l_j,\tau_j,L_j}{P_j}{j \in \mathcal J}$.   \textsc{P-Call} invokes a procedure, locally computing the next token and substituting the body of the procedure into the process. Rules \textsc{P-Compute}, \textsc{P-If}, and \textsc{P-Par} are standard.

The key novelty of out-of-order processes is the \textsc{P-Delay} rule, which allows a process to perform instructions in \emph{any} order, up to data- and control-dependencies. The latter implies processes cannot evaluate instructions nested within an $\PIfM{}$or $\BBranch$-expression.

\begin{figure}
\begin{subfigure}[t]{0.5\textwidth}
\begin{align*}
&\BuyItem_1(\proc b) =\\
&\qquad \PRecv{\itemID}{1}{\Tok};\\
&\qquad \PSet{\itemm}{sell(\itemID)};\\
&\qquad \PSend{\proc b}{3}{\Tok}{\itemm}\\
&\BuyItem_2(\proc s,\ \itemID) =\\
&\qquad \PSend{\proc s}{1}{\Tok}{\itemID};\\
&\qquad \PRecv{\itemm}{3}{\Tok}\\
&\seller[\PCall{\BuyItem_1}{\buyer_1}{4}{\tau_0};\\
&\qquad\PCall{\BuyItem_1}{\buyer_2}{5}{\tau_0}]\ |\\
&\buyer_1[\PCall{\BuyItem_2}{\seller, 123}{4}{\tau_0}]\ |\\
&\buyer_2[\PCall{\BuyItem_2}{\seller, 543}{5}{\tau_0}]
\end{align*}
\caption{}\label{fig:buy-item-epp}
\end{subfigure}\hfill
\begin{subfigure}[t]{0.5\textwidth}
\begin{align*}
&\val{StreamIt}_1(\proc c) =\\
&\qquad \PSend{\proc c}{1}{\Tok}{produce()};\\
&\qquad \PIfM{(itemsLeft() > 0)} \PThen\\
&\qquad\qquad \PSelect{\proc c}{4}{\Tok}{\small\textsc{More}};\  \PCall{\val{StreamIt}_1}{\proc c}{5}{\Tok}\\
&\qquad \PElse\ \PSelect{\proc c}{6}{\Tok}{\small\textsc{Done}}\\
&\val{StreamIt}_2(\proc p) =\\
&\qquad \PRecv{x}{1}{\Tok};\,\PSet{z}{consume(x)};\\
&\qquad \BBranch\,\{ (4,\Tok,{\small\textsc{More}}) \Rightarrow \PCall{\val{StreamIt}_2}{\p}{5}{\Tok},\\
&\qquad\ \ \ \ \,(6,\Tok,{\small\textsc{Done}}) \Rightarrow 0 \}\\
&\p_1[\PCall{\val{StreamIt}_1}{\proc c}{7}{\tau_0}]\ |\\
&\p_2[\PCall{\val{StreamIt}_1}{\proc c}{8}{\tau_0}]\ |\\
&\proc c[\PCall{\val{StreamIt}_2}{\proc p}{7}{\tau_0};\,\PCall{\val{StreamIt}_2}{\proc p}{8}{\tau_0}]
\end{align*}
\caption{}\label{fig:stream-it-epp}
\end{subfigure}
    \caption{Processes projected from \Cref{fig:chor-examples}.}
    \label{fig:proc-examples}
  \end{figure}

\subsection{Endpoint Projection}\label{sec:epp}

\Cref{fig:epp} defines the \emph{endpoint projection (EPP)} $\EPP{C}{}$ of a choreography $C$, translating it into a network. The rules follow from simple modifications to the textbook definition of EPP~\cite{montesi2023}. Projecting a conditional on a process that does not evaluate the guard uses the auxiliary partial operator $\Merge$, which produces a term that can react to the different branches by receiving different selections (this is standard).

\begin{figure}
    \centering\small
\[ \EPP{\mathscr C}{} = \bigcup_{i \in \mathcal I} \EPP{X_i(\vx \p, \vx {\p.x}) = C_i}{} \]
\[ \EPP{X_i(\vx \p, \vx{\p.x}) = C_i}{} = \{ X_{i,j}(\vx \p \setminus \p_j, \EPP{\vx{\p.x}}{p_j}) = \EPP{C_i}{\p_j}\ |\ \vx \p = \p_1,\dots,\p_n,\ j \le n \} \]
\[
\EPP{\CCom{\p}{e}{\q}{x}{\Line,t};\,C}{\r} =
\begin{cases}
    \PSend{\q}{\Line}{t}{e};\,\EPP{C}{\r} & \hbox{if $\r = \p$}\\
    \PRecv{x}{\Line}{t};\,\EPP{C}{\r} & \hbox{if $\r = \q$}\\
    \EPP{C}{\r} & \hbox{otherwise}
\end{cases}
\]
\[
\EPP{\CRecv{\p}{\q}{x}{\Line,t};\,C}{\r} =
\begin{cases}
    \PRecv{x}{\Line}{t};\,\EPP{C}{\r} & \hbox{if $\r = \q$}\\
    \EPP{C}{\r} & \hbox{otherwise}
\end{cases}
\]
\[
\EPP{\CPure{x}{\p}{e}{\Line,t};\,C}{\r} =
\begin{cases}
    \PSet{x}{\EPP{e}{\r}};\,\EPP{C}{\r} & \hbox{if $\r = \p$}\\
    \EPP{C}{\r} & \hbox{otherwise}
\end{cases}
\]
\[
\EPP{\CSelect{\p}{\q}{L}{\Line,t};\,C}{\r} =
\begin{cases}
    \PSelect{\q}{\Line}{t}{\EPP{e}{\r}};\,\EPP{C}{\r} & \hbox{if $\r = \p$}\\
    \PBranch{l,t,L}{\EPP{C}{\r}}{} & \hbox{if $\r = \q$}\\
    \EPP{C}{\r} & \hbox{otherwise}
\end{cases}
\]
\[
\EPP{\CRecvSelect{\p}{\q}{L}{\Line,t};\,C}{\r} =
\begin{cases}
    \PBranch{l,t,L}{\EPP{C}{\r}}{} & \hbox{if $\r = \q$}\\
    \EPP{C}{\r} & \hbox{otherwise}
\end{cases}
\]
\[
\EPP{\CIf{e}{\p}{C_1}{C_2}{\Line,t};\,C}{\r} =
\begin{cases}
    \PIf{\EPP{e}{\r}}{\EPP{C_1}{\r}}{\EPP{C_2}{\r}};\,\EPP{C}{\r} & \hbox{if $\r = \p$}\\
    \EPP{C_1}{\r} \Merge \EPP{C_2}{\r};\,\EPP{C}{\r} & \hbox{if $\r \in \pn(C_1,C_2) \setminus \p$}\\
    \EPP{C}{\r} & \hbox{otherwise}
\end{cases}
\]
\[
\EPP{\CCall{X_i}{\vx \p,\vx a}{\Line,t};\,C}{\r} =
\begin{cases}
    \PCall{X_{i,j}}{\vx \p \setminus \p_j,\EPP{\vx a}{\p_j}}{\Line}{t};\,\EPP{C}{\p_j} & \hbox{if $\r = \p_j$ where $\vx \p = \p_1,\dots,\p_n$}\\
    \EPP{C}{\r} & \hbox{otherwise}\\
\end{cases}
\]
\[
\EPP{\CCalling{\vx \q}{X_i}{\vx\p,\vx a}{C_1}{l,t};\,C_2}{\r} =
\begin{cases}
    \PCall{X_{i,j}}{\vx \p \setminus \p_j,\EPP{\vx a}{\p_j}}{\Line}{t};\,\EPP{C_2}{\p_j} & \hbox{if $\r \in \vx\q$ and $\r = \p_j$}\\
    \EPP{C_1;\,C_2}{r} & \hbox{if $\r \in \vx\p \setminus \vx\q$}\\
    \EPP{C_2}{\r} & \hbox{otherwise}\\
\end{cases}
\]
\begin{minipage}{0.4\textwidth}  
\[
\EPP{\Block{C_1};\,C_2}{\r} = \{\EPP{C_1}{\r}\}; \EPP{C_2}{\r}
\]
\[ \EPP{a_1,\dots,a_n}{\r} = \EPP{a_1}{\r},\dots, \EPP{a_n}{\r} \]
\[ \EPP{f(e_1,\dots,e_n)}{\r} = f(\EPP{e_1}{\r},\dots, \EPP{e_n}{\r}) \]
\end{minipage}%
\begin{minipage}{0.4\textwidth}  
    \[ \EPP{\At{v}{\p}}{\r} = 
    \begin{cases}
        v& \hbox{if $\r = \p$}\\
        \bot& \hbox{otherwise}
    \end{cases}\]
\[ \EPP{\p.x}{\r} = 
\begin{cases}
    x& \hbox{if $\r = \p$}\\
    \bot& \hbox{otherwise}
\end{cases}\]
\end{minipage}
\begin{multline*}
    \left(\PBranch{l_i,\tau_i,L_i}{P_i}{i \in \mathcal I}\right)\, \Merge\,
        \left(\PBranch{l_j,\tau_j,L_j}{P_j}{j \in \mathcal J} \right)
    = \PBranch{l_k,\tau_k,L_k}{P_k}{k \in \mathcal I \cup \mathcal J}
    \\ \hbox{if $\Disjoint{\{L_i : i \in \mathcal I\}}{\{L_j : j \in \mathcal J\}}$}
\end{multline*}
    \caption{Endpoint projection}
    \label{fig:epp}
\end{figure}

\Cref{fig:proc-examples} shows networks projected from the choreographies of \Cref{fig:chor-examples}. Notice the choreographic procedures \emph{BuyItem} and \emph{StreamIt} are each split into two process procedures---one for each role. Communications in the choreography are, as usual, projected into send and receive instructions. Conditionals in the choreography are projected into an $\mathsf{if}$-instruction at one process and a $\mathsf{branch}$-instruction at the other processes awaiting its decision.

\begin{figure}[t]
  \centering
\begin{minipage}{0.98\textwidth}
\begin{align*}
	&0 \Extends 0&&\\
	&(P_1;\,P_2) \Extends (Q_1;\,Q_2) &&\hbox{if $P_i \Extends Q_i$ for $i=1,2$}\\
	&(\PIf{e}{P_1}{P_2}) \Extends (\PIf{e}{Q_1}{Q_2}) &&\hbox{if $P_i \Extends Q_i$ for $i=1,2$}\\
	&I_1 \Extends I_2 &&\hbox{if $I_1 = I_2$ or $I_1 = I_1 \Merge I_2$}
\end{align*}
\end{minipage}
      \begin{minipage}{0.38\textwidth}
        \begin{align*}
            &\keys(0) = \epsilon\\
            &\keys(I;\,P) = \keys(I), \keys(P)\\
            &\keys(\PSend{\p}{l}{t}{e}) = (l,t)\\
            &\keys(\PRecv{x}{l}{t}) = (l,t)\\
            &\keys(\PSet{x}{e}) = \epsilon\\
            &\keys(\PSelect{\p}{\Line}{t}{L}) = (l,t)
        \end{align*}
    \end{minipage}%
    \begin{minipage}{0.6\textwidth}
        \begin{align*}
            &\keys(\PBranch{l_i,\tau_i,L_i}{P_i}{i \in \mathcal I}{}{}) =\\
            &\qquad[(l_i,\tau_i)\ |\ i \in \mathcal I], [\keys(P_i)\ |\ i \in \mathcal I]\\
            &\keys(\PIf{e}{P_1}{P_2}) = \keys(P_1), \keys(P_2)\\
            &\keys(\PCall{X}{\vx\p,\vx a}{l}{t}) = (l,t)\\
            &\keys(\{P_1\};\, P_2) = \keys(P_1), \keys(P_2)
        \end{align*}
    \end{minipage}
\caption{Auxiliary definitions for the EPP Theorem ($\Extends$ and $\keys$)}
\label{fig:proc-keys}
\end{figure}

Below we formulate the hallmark \emph{EPP Theorem}, which states that a choreography $C$ and its projection $\EPP{C}{}$ evolve in lock-step, up to the usual $(\Extends)$ relation from Montesi~\cite{montesi2023} (given in \cref{fig:proc-keys}). Importantly, we update the theorem to restrict our attention to \emph{well-formed} networks and choreographies. We say that a network $N$ is well-formed if the keys in each process are distinct, i.e., $\keys(P)$ is distinct for each $\p[P]$ in $N$ ($\keys(P)$ is given in \Cref{fig:proc-keys}). The restriction allows us to ignore processes such as
$\p[\PBranch{1,\tau,L}{P_1}{};\,\PBranch{1,\tau,L}{P_2}{}]$,
which could only be projected from a choreography where two distinct instructions have the same integrity key $(1,\tau)$. \changed{}{This leads to the following lemma:}

\begin{lemma}\label{lem:branching-keys}
  Let $C,Q$ be well-formed. If $Q \Extends \EPP{C}{\q}$ then $\keys(Q) \supseteq \keys_\q(C)$, where
$$
\keys_\q(C) = [ (l,t)\ |\ (\CRecv{\p}{\q}{x}{l,t}) \in \stats(C) ],
    [ (l,t)\ |\ (\CCom{\p}{e}{\q}{x}{l,t}) \in \stats(C)].
$$
\end{lemma}

\changed{}{The key difficulty of proving the EPP Theorem was finding the right definition of well-formedness (\Cref{thm:preservation,thm:integrity}). With the definition established, the entire proof follows directly from textbook induction principles (\emph{c.f.}~\cite{montesi2023}). We sketch the proof in \Cref{sec:epp-theorem-proof}.}

\begin{restatable}[EPP Theorem]{theorem}{epp}\label{thm:epp} Let $\CCfg{C}{\Sigma}{K}$ be a well-formed configuration.
\begin{enumerate}
    \item (Completeness) If $\CCfg{C}{\Sigma}{K} \CStep{\p} \CCfg{C'}{\Sigma'}{K'}$ then $\PCfg{\EPP{C}{}}{\Sigma}{K} \PStep{\p} \PCfg{N'}{\Sigma'}{K'}$ for some well-formed $N'$ where $N' \Extends \EPP{C'}{}$.
    \item (Soundness) If $\PCfg{N}{\Sigma}{K} \PStep{\r} \PCfg{N'}{\Sigma'}{K'}$ for some well-formed $N$ where $N \Extends \EPP{C}{}$, then $\CCfg{C}{\Sigma}{K} \CStep{\p} \CCfg{C'}{\Sigma'}{K'}$ for some $C'$ where $N' \Extends \EPP{C'}{}$.
\end{enumerate}
\end{restatable}

\section{A Non-Blocking Communication API for Choral}\label{sec:choral}

In this section, we show how the ideas in $O_3$ can be applied in practice. To this end, we consider Choral~\cite{giallorenzo2023}: a state-of-the-art choreographic programming language based on Java\@. Choral is designed to support real-world programming and interoperate with Java, so it is much more sophisticated than our minimalistic theory. Data locations in Choral are lifted to the type level and communication is expressed by invoking methods of \emph{channel} objects.

Choral's intended programming model consists of sequential processes that block to receive messages. However, to improve performance programmers can use Java's \texttt{CompletableFuture} API, thereby introducing intraprocess concurrency and out-of-order execution. This breaks the programming model and introduces CIVs (cf.~\Cref{sec:motivation}) that could cause crashes or silent memory corruption. Motivated by our formal model, we developed \emph{Ozone}: an API for Choral programmers to safely mix choreographies with futures.
In the remainder of this section, we introduce Choral and Ozone and we show how programmers can mix choreographies with futures achieve significant speedups in practical applications.

\begin{figure}[t]
    \begin{minted}[linenos,fontsize=\footnotesize]{'choral.py:ChoralLexer -x'}
public class ConcurrentSend@(KS, CS, S, C) {
    public void start(
        String@KS key, String@CS txt, Client@C client,
        Token@(KS, CS, S, C) tok,
        AsyncChannel@(KS, S) ch1, AsyncChannel@(CS, S) ch2, AsyncChannel@(S, C) ch3
    ) {
        // Services send data to the server.
        CompletableFuture@S keyS = ch1.fcom(key, 1@(KS,S), tok);
        CompletableFuture@S txtS = ch2.fcom(txt, 2@(CS,S), tok);
    
        // Server forwards data to the client.
        ch3.fcom(keyS, 3@(S,C), tok)
            .thenAccept(client::decrypt);
        ch3.fcom(txtS, 4@(S,C), tok)
            .thenAccept(client::display);
    }
}
    \end{minted}
    \caption{An implementation of the choreography in \Cref{fig:forwarding-data} using Choral and the Ozone API.}\label{fig:choral-concurrent-send}
\end{figure}

\subsection{Concurrent Messages}\label{sec:choral-concurrent-send}

We introduce the Ozone API with an implementation of the choreographic
procedure from \Cref{fig:forwarding-data}. The implementation is shown in \Cref{fig:choral-concurrent-send}, which
defines a class called \texttt{ConcurrentSend} parameterized by four
roles (i.e.~process parameters): \texttt{KS}, \texttt{CS}, \texttt{S}, and \texttt{C}. In this
class, the \texttt{start} method implements the procedure itself. As in
our formal model, the procedure is parameterized by distributed data: On
line 3, parameter \texttt{key} is a \texttt{String} located at
\texttt{KS}; \texttt{txt} is a \texttt{String} located at \texttt{CS};
and \texttt{client} is a \texttt{Client} object at \texttt{C},
representing the client's user interface. The \texttt{start} procedure
is also parameterized by session tokens, which we introduced in \Cref{fig:choral-concurrent-send}, on line 4. The parameter \texttt{Token@(KS, CS, S, C)\ tok} is
syntactic sugar for the parameter \emph{list}
\texttt{Token@KS tok\_KS, ..., Token@C tok\_C}.\footnote{This syntactic sugar is provided for readability and is not currently supported by the Choral compiler. We will also use syntactic sugar for lambda expressions and omit obvious type annotations later in this section. Our actual implementation uses desugared versions of the syntax.}
The last three parameters on line 5 are \emph{channels}. In Choral,
channels are used to communicate data from one role to another. If
\texttt{ch} is a channel of type
\texttt{Channel@(A,B)\textless{}T\textgreater{}} and \texttt{e} is an
expression of type \texttt{T@A}, then the expression \texttt{ch.com(e)}
is a communication that produces a value of type \texttt{T@B}.

Our main contribution in the Ozone API is a custom channel
\texttt{AsyncChannel@(A,B)\textless{}T\textgreater{}} with a method
\texttt{fcom} for safely communicating data with non-blocking semantics.
The \texttt{fcom} method is similar to \texttt{com}, but with the following
differences:
\begin{itemize}
\item
  Whereas \texttt{com} takes one argument, \texttt{fcom} takes three: a
  payload, a line number, and a session token. The latter two arguments
  form an integrity key, of which both the sender and
  receiver have a copy.
\item
  When the receiver \texttt{B} executes a \texttt{com} instruction, its
  thread becomes blocked until the value (of type \texttt{T@B}) has been
  delivered. In contrast, \texttt{fcom} creates a Java \emph{future} (of
  type \texttt{CompletableFuture@B\textless{}T\textgreater{}}) which is
  a placeholder at \texttt{B} that will hold a value of type \texttt{T}
  once the message is delivered. Instead of blocking, \texttt{fcom}
  immediately returns that future to the calling thread. The thread can
  then assign a callback to handle the message and proceed with other
  useful work.
\end{itemize}
Lines 8 and 9 of \Cref{fig:choral-concurrent-send} show \texttt{fcom} being used to transport
\texttt{key} and \texttt{txt} to the server \texttt{S}. The expression
\texttt{1@(KS,\ S)} is sugar for the list \texttt{1@KS,\ 1@S} and we
assume the replicated value \texttt{tok} is expanded into the list
\texttt{tok\_KS,\ tok\_S}. Thus both sender and receiver pass integrity
keys as arguments to \texttt{fcom}.

Lines 12-15 of \Cref{fig:choral-concurrent-send} show how the server \texttt{S} and client
\texttt{C} use the future values. On line 12, the server uses an
overloaded version of \texttt{fcom} that takes
\texttt{CompletableFuture@S} instead of \texttt{T@S}. The method assigns
to the future a callback, which forwards the result to the client once
the future has been completed. The result of \texttt{fcom} on line 12 is
a \texttt{CompletableFuture@C}, to which the client binds a callback on
line 13: when the key from \texttt{S} finally arrives at \texttt{C}, the
client will proceed to invoke the method \texttt{client.decrypt} with
the key as an argument. Lines 14 and 15 do the same, but with the value
of \texttt{txt}. As we will see below, the values of \texttt{key} and
\texttt{txt} can arrive at the client in any order, so the callbacks on
lines 13 and 15 can execute in any order---even in parallel.

\begin{figure}
    \small
\begin{minipage}{0.5\textwidth}
\begin{minted}[linenos,fontsize=\footnotesize]{'choral.py:ChoralLexer -x'}
public class ConcurrentSend_KS {
  public void start(
    String key, Token tok_KS, 
    AsyncChannel ch1
  ) {
    ch1.fcom(key, 1, tok_KS);
  }
}
public class ConcurrentSend_S {
  public void start(
    Token tok_S, AsyncChannel ch1, 
    AsyncChannel ch2, AsyncChannel ch3
  ) {
    CompletableFuture keyS = 
      ch1.fcom(1, tok_S);
    CompletableFuture txtS = 
      ch2.fcom(2, tok_S);

    ch3.com(keyS, 3, tok_S);
    ch3.com(txtS, 4, tok_S);
  }
}
\end{minted}
\end{minipage}\hfill%
\begin{minipage}{0.45\textwidth}
\begin{minted}[linenos,fontsize=\footnotesize,firstnumber=last]{'choral.py:ChoralLexer -x'}
public class ConcurrentSend_CS {
  public void start(
    String txt, Token tok_CS, 
    AsyncChannel ch2
  ) {
    ch2.fcom(txt, 2, tok_CS);
  }
}

public class ConcurrentSend_C {
  public void start(
    Client client, Token tok_C, 
    AsyncChannel ch3
  ) {
    ch3.fcom(3, tok_C)
       .thenAccept(client::decrypt);
    ch3.fcom(4, tok_C)
       .thenAccept(client::display);
  }
}
\end{minted}
\end{minipage}
\caption{Endpoint projection of \Cref{fig:choral-concurrent-send}.}\label{fig:choral-concurrent-send-epp}
\end{figure}

\subsubsection{Endpoint projection}\label{sec:choral-epp}

By running the Choral compiler, \texttt{ConcurrentSend@(KS,CS,S,C)} is
\emph{projected} to generate four Java classes, shown in \Cref{fig:choral-concurrent-send-epp}. Each
class implements the behavior of its corresponding role. For example,
\texttt{ConcurrentSend\_KS} implements the behavior of \texttt{KS}. Its
\texttt{start} method is parameterized by: \texttt{key}, which
corresponds to the \texttt{key} in \Cref{fig:choral-concurrent-send}; \texttt{tok\_KS}, the copy of
the token \texttt{tok} belonging to \texttt{KS}; and \texttt{ch1}, a
channel endpoint that connects \texttt{KS} to \texttt{S}. Following the
reasoning in \Cref{fig:epp}, these behaviors will not exhibit deadlocks or
communication integrity errors when composed (assuming the implementations of Choral and Ozone are correct).

Let us see how integrity keys prevent CIVs in \Cref{fig:choral-concurrent-send-epp}. Notice that the
Choral instruction
\texttt{CompletableFuture@S\ keyS\ =\ ch1.fcom(key,\ 1@(KS,S),\ tok);}
on line 8 of \Cref{fig:choral-concurrent-send} is projected into two instructions:
\begin{enumerate}
\item
  \texttt{ch1.fcom(key,\ 1,\ tok\_KS)} at the sender \texttt{KS}; and
\item
  \texttt{CompletableFuture\ keyS\ =\ ch1.fcom(1,\ tok\_S)} at the
  receiver \texttt{S}.
\end{enumerate}
The former instruction is parameterized by a payload and an integrity
key and produces nothing. The latter instruction is parameterized only
by an integrity key (with no payload) and produces a future. When
\texttt{KS} sends \texttt{key} to \texttt{S}, it combines the payload
with integrity key \texttt{(1,\ tok\_KS)}. Dually, \texttt{S} creates a
future that will only be completed when a message with the integrity key
\texttt{(1,\ tok\_S)} is received. Since \texttt{tok\_KS} and
\texttt{tok\_S} have the same value, the send- and receive-operations
are guaranteed to match.

On lines 14-17 of \Cref{fig:choral-concurrent-send-epp}, the server \texttt{S} sets listeners for
\texttt{key} and \texttt{txt}. On lines 19-20, \texttt{S} schedules the
values to be forwarded to \texttt{C}; notice that even with FIFO
channels, \texttt{key} and \texttt{txt} may arrive in any order.
Consequently, \texttt{S} may forward their values to \texttt{C} in any
order. On lines 37-40 of \Cref{fig:choral-concurrent-send-epp}, the client creates futures to hold the
values of \texttt{key} and \texttt{txt} and sets callbacks to be invoked
when the values arrive. Here we see the importance of integrity keys:
the client uses \texttt{(3,\ tok\_C)} and \texttt{(4,\ tok\_C)} to
disambiguate the \texttt{key} message from the \texttt{txt} message.
Without integrity keys, mixing Choral choreographies with Java Futures
would be unsafe. As shown in \Cref{sec:epp}, our solution is correct even
when the underlying transport protocol can deliver messages
out of order.

\begin{figure}
    \begin{minted}[linenos,fontsize=\footnotesize]{'choral.py:ChoralLexer -x'}
public class ConcurrentClients@(KS, CS, S, C1, C2) {
    public void start(
        AsyncChannel@(KS, S) ch1, AsyncChannel@(CS, S) ch2,
        AsyncChannel@(S, C1) ch3, AsyncChannel@(S, C2) ch4,
        KeyService@KS keyService, ContentService@CS contentService,
        Client@C1 client1, String@(KS, CS) clientID1,
        Client@C2 client2, String@(KS, CS) clientID2,
        Token@(KS, CS, S, C1, C2) tok
    ) {
        (new ConcurrentSend2()).start(ch1, ch2, ch3,
            keyService.getKey(clientID1), contentService.getContent(clientID1),
            client1, tok.nextToken( 0@(KS,CS,S,C1) ));

        (new ConcurrentSend2()).start(ch1, ch2, ch4,
            keyService.getKey(clientID2), contentService.getContent(clientID2),
            client2, tok.nextToken( 1@(KS,CS,S,C2) ));
    }
}
    \end{minted}
    \caption{A Choral choreography invoking \texttt{ConcurrentSend2}.}\label{fig:choral-procedures}
\end{figure}

\begin{figure}
\begin{minipage}{0.55\textwidth}
\begin{minted}[linenos,fontsize=\footnotesize]{'choral.py:ChoralLexer -x'}
public class ConcurrentClients_KS {
  public void start( ... ) {
    (new ConcurrentSend2_KS()).start(ch1,
        keyService.getKey(clientID1),
        tok.next(0));

    (new ConcurrentSend2_KS()).start(ch1,
        keyService.getKey(clientID2),
        tok.next(1));
  }
}
\end{minted}
\end{minipage}
\begin{minipage}{0.45\textwidth}
\begin{minted}[linenos,fontsize=\footnotesize,firstnumber=last]{java}
public class ConcurrentClients_S {
  public void start( ... ) {
    (new ConcurrentSend2_S()).start(
        ch1, ch2, ch3, tok.next(0));

    (new ConcurrentSend2_S()).start(
        ch1, ch2, ch3, tok.next(1));
  }
}
\end{minted}
\end{minipage}
\caption{Endpoint projection of \Cref{fig:choral-procedures} (representative examples).}
\end{figure}

\subsection{Procedure calls}\label{sec:choral-procedures}

\Cref{sec:choral-concurrent-send} showed how the line numbers in an integrity key could prevent
CIVs. We now briefly show how the \emph{tokens} in an integrity key
prevent \emph{interprocedural} CIVs. \Cref{fig:choral-procedures} depicts a choreography that
invokes two instances of \texttt{ConcurrentSend2}: the first instance
with client \texttt{C1}, and the second instance with client
\texttt{C2}. On lines 12 and 16, the roles all compute fresh tokens for
each procedure they're involved in, like in our formal model; the syntax
\texttt{tok.nextToken(\ 0@(KS,CS,S,C1)\ )} is sugar for
\texttt{tok\_KS.nextToken(0@KS),\ ...,\ tok\_C1.nextToken(0@C1)}, and
the method \texttt{t.nextToken(l)} implements the function
$\NextToken(l,t)$. These fresh tokens ensure that, even if messages
from \texttt{KS} to \texttt{S} are delivered out of order
there is no chance that messages
from the first procedure invocation will be confused for messages from
the second invocation.

\begin{figure}[t]
  \centering
    \begin{subfigure}{0.5\textwidth}
        \includegraphics[width=\textwidth]{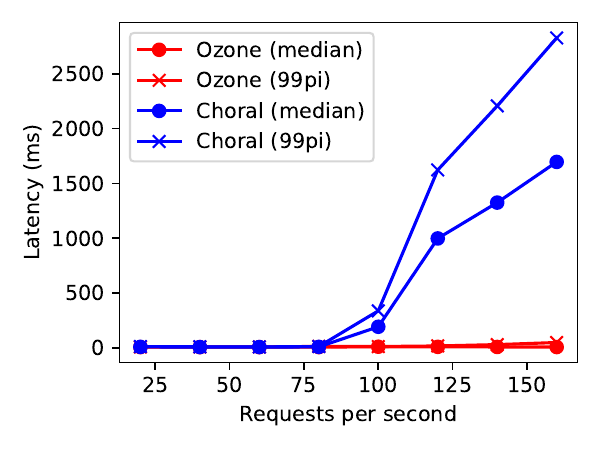}
        \subcaption{Concurrent producers latency (lower is better)}\label{fig:data-producers}
    \end{subfigure}\hfill
    \begin{subfigure}{0.5\textwidth}
      \centering
        \includegraphics[width=0.8\textwidth]{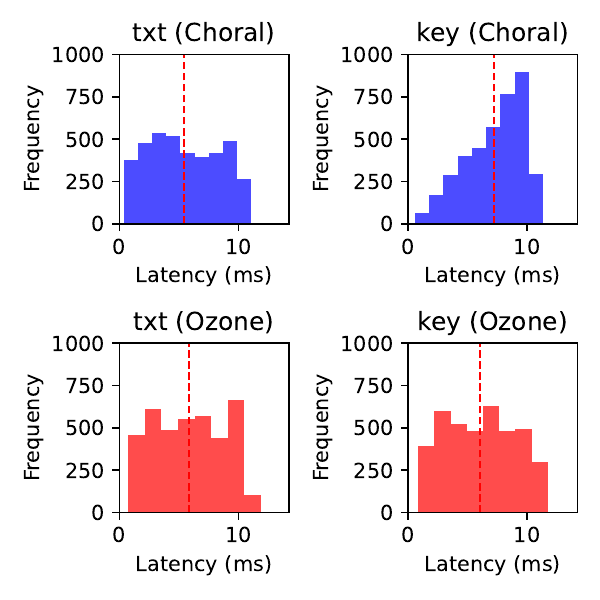}
        \subcaption{Concurrent senders latency (further left is better)}\label{fig:data-senders}
    \end{subfigure}
    \caption{Microbenchmark}%
    \label{fig:microbenchmarks}
  \end{figure}

\begin{figure}[t]
  \centering
  \begin{subfigure}{0.45\textwidth}
    \includegraphics[width=\textwidth]{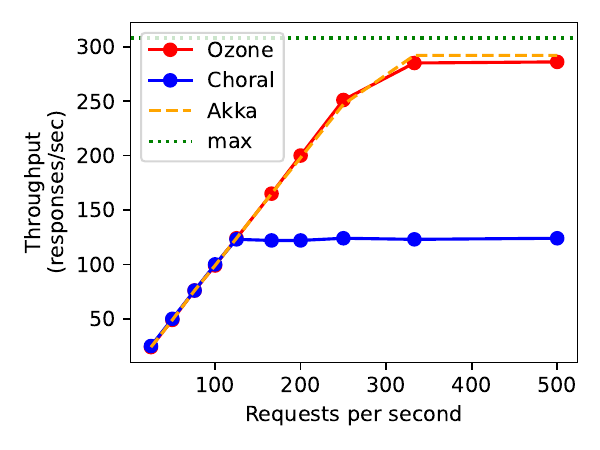}
    \subcaption{Throughput (higher is better)}%
    \label{fig:modelserving-throughput}
  \end{subfigure}\hfill%
  \begin{subfigure}{0.45\textwidth}
    \includegraphics[width=\textwidth]{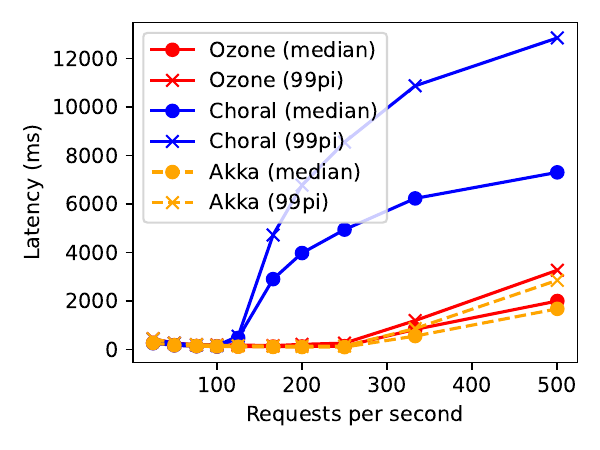}
    \subcaption{Latency (lower is better)}%
    \label{fig:modelserving-latency}
  \end{subfigure}%
  \caption{Model serving}%
  \label{fig:modelserving}
\end{figure}

\begin{figure}
  \centering
  \begin{minipage}{0.48\textwidth}
  \includegraphics[width=\textwidth]{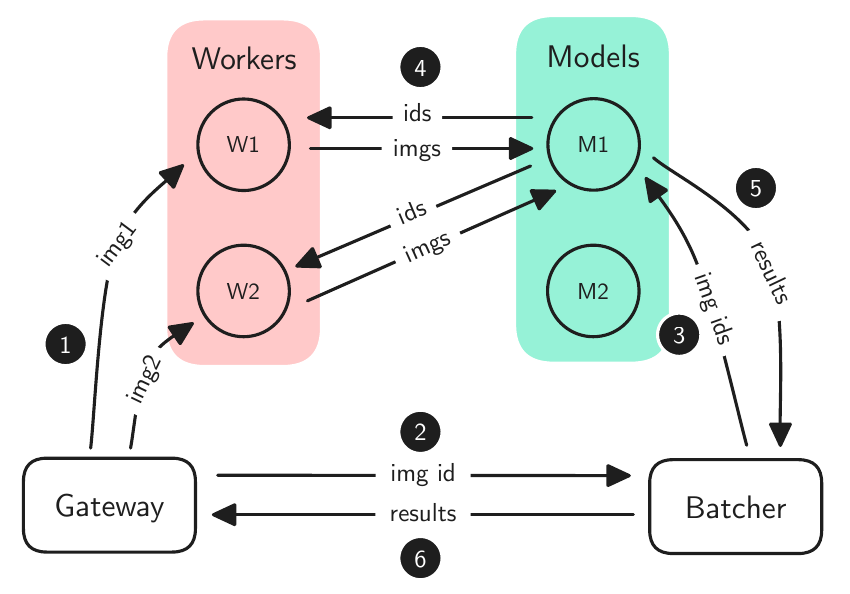}
  \captionof{figure}{Architecture for the image classification pipeline.}%
  \label{fig:modelserving-architecture}
\end{minipage}\hfill
\begin{minipage}{0.48\textwidth}
 \includegraphics[width=\textwidth]{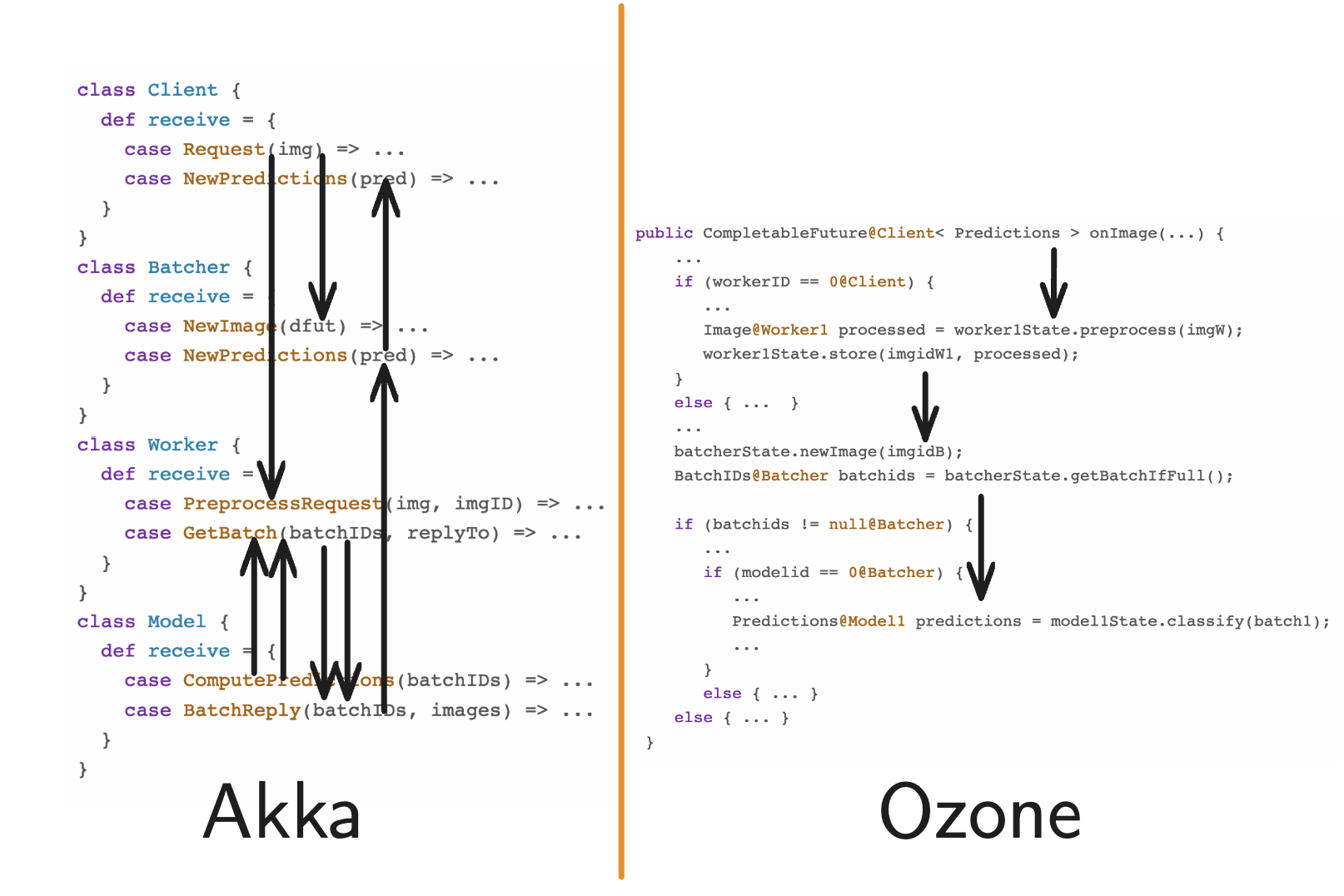}
  \captionof{figure}{Control flow comparison between hand-written Akka processes and Ozone.}%
  \label{fig:bar}
\end{minipage}%
\end{figure}

\subsection{Evaluation}\label{sec:evaluation}

We evaluated Ozone with microbenchmarks based on
\Cref{fig:concurrent-producers,fig:forwarding-data} and with a
model serving benchmark from Wang et al~\cite{wang21}. The experiments were
carried out on \changed{}{a six-node Linux cluster, with two Intel Xeon Gold 6130 CPUs and 384 GB of memory per node and an average bandwidth of 15 Gbps.}

The first microbenchmark is a version of \Cref{fig:concurrent-producers} from the introduction. Each producer iteratively invokes the choreography at a fixed rate and, in response to each request, the server simulates computation by sleeping for 0--5 milliseconds. \Cref{fig:data-producers} shows the median and 99th percentile latency for server responses to worker requests. In the Choral implementation, the server quickly becomes a bottleneck because of its fixed processing order: a request from \(\p_1\) must be handled before a request from \(\p_2\), and both requests must be handled for the \(i\)-th iteration before they can be handled for the \((i+1)\)-th iteration. In the Ozone implementation, requests from different producers can be handled out of order and producers can start a new instance of the choreography without waiting for the second one to complete. Consequently the server spends less time waiting for requests, so it can handle much higher request rates.

The second microbenchmark is a version of \Cref{fig:forwarding-data} from \Cref{sec:motivation}, in which the server sends messages to the microservices \emph{ks} and \emph{cs} and forwards their responses to the client. Each microservice takes 0--5 milliseconds to compute its response. The latency histogram for the Choral implementation (top) shows how the time for the client to receive \texttt{txt} depends on the time to compute \texttt{txt}, but the time to receive \texttt{key} depends on \emph{both} the time to compute \texttt{txt} and the time to compute \texttt{key}. In contrast, the Ozone implementation (bottom) allows the server to forward \texttt{key} to the client without waiting for \texttt{txt}---thereby reducing the average latency for \texttt{key} by more than 30\%.

To measure the impact of Ozone on a realistic application, we ported the image classification pipeline of Wang et al~\cite{wang21} to Choral (\Cref{fig:modelserving-architecture}). In this pipeline, images are received by a Gateway that performs load balancing and forwards the images to a pair of Worker services for preprocessing. A Batcher service collects requests and sends them as a batch to a Model service, which fetches the processed images and performs image classification. This architecture allows applications to harness \emph{intra}-GPU parallelism by increasing the batch size (at the cost of latency) and \emph{inter}-GPU parallelism by increasing the number of Model services.
\Cref{fig:modelserving} shows the performance for Choral and Ozone implementations of the pipeline, using sleeps to simulate computation. The plots also show the performance of an implementation in the \emph{Akka} actor framework~\cite{akka} \changed{}{and the theoretical maximum throughput of the Model services.}

The Choral implementation has a bottleneck: after the Batcher sends work to a Model, it waits for the Model's response and becomes blocked. In the Ozone implementation, the Batcher binds the Model service's response to a \texttt{CompletableFuture} and continues receiving requests from the Gateway. Consequently, the throughput and latency for the Ozone implementation can scale with the number of requests until both Models become saturated with work. \changed{}{\Cref{fig:modelserving} shows our library scales similarly to hand-written reactive processes in Akka, though the latter perform slightly better under high load because the Akka framework is heavily optimized to handle network congestion; these same optimizations can be applied to Choral, but they are orthogonal to our present work.} We conclude that our methodology can achieve good performance while providing the benefits of choreographic programming: (i) absence of bugs like deadlocks and mismatched communications (e.g., sending a message with the wrong type or at the wrong time)~\cite{giallorenzo2023,montesi2023,lugovic2023}; (ii) and improved readability, since control flow is easier to follow in choreographies than in processes (see \Cref{fig:bar}), as discussed in~\cite{shen2023a,KSZK23,lugovic2023,giallorenzo2023}.

\section{Related Work}\label{sec:related-work}


In early choreographic languages, the sequencing operator \(I;\ C\) had strict sequential semantics; concurrency could only be introduced via an explicit parallel operator \(C\, ||\, C'\)~\cite{qiu2007a,lanese2008,carbone2012a}.
Explicit parallelism was later replaced by a relaxed sequencing operator \(I;\,C\) that would allow
instructions in \(C\) to be evaluated before \(I\) under certain
conditions~\cite{carbone2013}. This presents the benefits of offering a simple syntax for choreographies and, at the same time, automatically inferring what can be safely executed concurrently. For these reasons, relaxed sequencing has been adopted in most recent works on choreographic programming (e.g.,~\cite{hirsch2022,giallorenzo2023,GHM24,CMP23,JB22,KSZK23}) and its textbook presentation~\cite{montesi2023}.
Our present work makes the sequencing operator
even more relaxed, allowing all instructions to be executed out of
order, up to data- and control-dependency.
\changed{To the best of our knowledge, our choreography model is the first to allow non-FIFO communication between processes.}{To the best of our knowledge, our model is the first to support non-FIFO communication in the setting of choreographic programming.}

Our work is closely related to choreographic \emph{multicoms}:
sets of communications that can be executed out of order, up to data
dependency~\cite{cruz-filipe2018}. However, multicoms do not
allow \emph{computation} to be performed out of order, as in \Cref{fig:concurrent-producers-ideal}. Multicoms therefore do not need to address the communication integrity problem, which we focus on in
this work. Relatedly, previous work investigated modeling asynchronous communication by making send actions non-blocking~\cite{carbone2013,honda2016,cruz-filipe2017a,pohjola2022,hirsch2022,montesi2023}, but none of considered non-blocking receive. Thus, they are not expressive enough to capture the behaviors that we are interested in here.


In terms of expressivity, there is some overlap between our model and \emph{nondeterministic choreographies}~\cite{montesi2023}, which use an explicit \emph{choreographic choice} operator \(C +_{\mathsf p} C'\). Nondeterministic choreographies can implement the execution in \Cref{fig:concurrent-producers-ideal} with:
$$
\begin{pmatrix}
    \BCom{\buyer_1}{id}{\seller}{id_1};\\
    \BCom{\buyer_2}{id}{\seller}{id_2};\\
    \dots
\end{pmatrix}
+_\seller
\begin{pmatrix}
    \BCom{\buyer_2}{id}{\seller}{id_2};\\
    \BCom{\buyer_1}{id}{\seller}{id_1};\\
    \dots
\end{pmatrix}
$$
\Cref{fig:forwarding-data} can also be expressed with nondeterministic choreographies:
$$
\begin{pmatrix*}[l]
    1: \BCom{\proc{cs}}{getText()}{\proc s}{txt};\\
    2: \BSelect{\proc s}{\proc c}{\small\textsc{TxtFirst}};\\
    3: \BCom{\proc s}{txt}{\proc c}{txt};\\
    4: \proc c.display(\proc c.txt);\\
    5: \BCom{\proc{ks}}{getKey()}{\proc s}{key};\\
    6: \BCom{\proc s}{key}{\proc c}{key};\\
    7: \proc c.decrypt(\proc c.key)
\end{pmatrix*}
+_\proc{s}
\begin{pmatrix*}[l]
    \ 8: \BCom{\proc{ks}}{getKey()}{\proc s}{key};\\
    \ 9: \BSelect{\proc s}{\proc c}{\small\textsc{KeyFirst}};\\
    10: \BCom{\proc s}{key}{\proc c}{key};\\
    11: \proc c.decrypt(\proc c.key);\\
    12: \BCom{\proc{cs}}{getText()}{\proc s}{txt};\\
    13: \BCom{\proc s}{txt}{\proc c}{txt};\\
    13: \proc c.display(\proc c.txt)
\end{pmatrix*}
$$
\changed{}{Compared to $O_3$, these implementations are much larger because they require programmers to statically encode all desired schedules. One can easily forget to include some schedules or encode them incorrectly: for instance, if one moved line 12 up to line 10 above, it would eliminate the extra concurrency that was gained by receiving the messages out of order. Thus, our approach is more robust and simpler for the programmer.}

On the other hand, nondeterministic choreographies can express some programs that our model cannot. For example, choreographic choice can assign different variable names to messages, according to their arrival order.
Other choreographic languages include nondeterministic operators~\cite{lanese2008,bravetti2008}, but they do not support computation (a requirement for choreographic programming) or recursion.




Choral is arguably the most powerful implementation of choreographic programming to date, but there are also others that target, e.g., Haskell, Java, Jolie, and Rust~\cite{carbone2013,CMMP22,DGGLM16,shen2023a,KSZK23}.
We believe that implementing the out-of-order semantics of $O_3$ in these languages is possible, too. However, it would likely require more work that touches also the implementation of EPP because, differently from Choral, these languages do not support user-defined communication primitives.
Since Choral is more expressive than all other current choreographic programming languages, we have targeted the most general case.

In~\cite{lugovic2023}, the authors introduce a Choral library for handling protocols that might deliver messages out of order. Unlike Ozone, this library requires explicitly writing which parts of a choreography are independent. Dependencies between actions also need to be managed at a low level via side-effects. In our approach, out-of-order communications can be elegantly combined by using futures. Furthermore, the work in~\cite{lugovic2023} does not deal with CIVs (which might arise if programmers are not careful) and presents no formal model.

A more loosely related line of research is that on multiparty session types (MPSTs)~\cite{honda2016}, where abstract choreographies without data or computation are used as protocol specifications \changed{}{that are compiled to ``local session types''}. Similarly to most work on choreographic programming, some works on MPSTs allow for non-blocking send, but not non-blocking receive as in $O_3$.
\changed{}{Previous work considered reordering actions in local session types~\cite{cutner22}, but these reorderings are necessarily limited because asynchronous multiparty session subtyping is undecidable in general~\cite{bravetti21}.} Interprocedural MPSTs have been presented~\cite{DH12}, but unlike $O_3$, the procedure calls require a central coordinator. Similar comments hold for recent investigations that add nondeterminism because of crashes~\cite{APN17,VHEZ21}. More generally, it is unclear whether ``concurrency up to data dependency'' could be expressed with MPSTs in their current form, since the types do not encode data dependencies.

\section{Conclusion}\label{sec:conclusion}\label{sec:discussion}

We investigated a model for choreographic programming in which processes can execute out of order and messages can be reordered by the network. These features improve the performance of choreographies, without requiring programmers to rewrite their code, by allowing processes to better overlap communication with computation. However, compilers that use these features must have mechanisms in place to prevent communication integrity violations (CIVs). We presented a scheme to prevent CIVs by attaching dynamically-computed integrity keys to each message.
Our results enlarge the class of behaviors that can be captured with choreographic programming without renouncing its correctness guarantees.

An important subject for future work is confluence. Statements can read and write to the local state of a process, so executing statements out of order can cause nondeterminism. Sometimes this nondeterminism is desirable (for instance, to implement consensus algorithms) but sometimes the nondeterminism is unexpected and causes bugs. In our formal model, nondeterminism could be controlled manually by allowing programmers to insert synthetic data dependencies. For example, below we use a hypothetical keyword $\mathsf{barrier}_\p$ to prevent a file from being closed before it has been written-to:
\[\BPure{file}{\p}{open(\hbox{``foo.txt''})};\ 
\p.write(\p.file, \hbox{``hello''});\ 
\mathsf{barrier}_\p;\ 
\p.close(\p.file)\]
More generally, future work could develop a static analysis that identifies when two statements are not safe to execute out of order.

Another opportunity for static analysis to improve on our work concerns the size of session tokens. We chose to represent session tokens as lists of integers, which allowed processes to compute new session tokens without coordinating with one another. However, this encoding means the size of a token is proportional to the depth of the call stack---a problem for tail-recursive programs such as \emph{StreamIt} in \Cref{fig:stream-it}. Fortunately, it is easy to see that communication integrity in \emph{StreamIt} could be achieved in constant space by representing the token as a single integer, incremented upon each recursive call, assuming that processes do not participate in multiple instances of the choreography concurrently. With static analysis, a compiler could identify such programs and use a more efficient session token representation.

\bibliography{main}

\begin{thebibliography}{10}

\bibitem{APN17}
Manuel Adameit, Kirstin Peters, and Uwe Nestmann.
\newblock Session types for link failures.
\newblock In Ahmed Bouajjani and Alexandra Silva, editors, {\em Formal
  Techniques for Distributed Objects, Components, and Systems - 37th {IFIP}
  {WG} 6.1 International Conference, {FORTE} 2017, Held as Part of the 12th
  International Federated Conference on Distributed Computing Techniques,
  DisCoTec 2017, Neuch{\^{a}}tel, Switzerland, June 19-22, 2017, Proceedings},
  volume 10321 of {\em Lecture Notes in Computer Science}, pages 1--16.
  Springer, 2017.
\newblock \href {https://doi.org/10.1007/978-3-319-60225-7\_1}
  {\path{doi:10.1007/978-3-319-60225-7\_1}}.

\bibitem{agha1990a}
Gul Agha.
\newblock {\em {{ACTORS}} - a Model of Concurrent Computation in Distributed
  Systems}.
\newblock {{MIT Press}} Series in Artificial Intelligence. {MIT Press},
  {Cambridge, MA}, 1990.

\bibitem{akka}
Akka.
\newblock {\url{https://akka.io/}}, 2024.

\bibitem{baker1977}
Henry~C. Baker and Carl Hewitt.
\newblock The incremental garbage collection of processes.
\newblock In James Low, editor, {\em Proceedings of the 1977 Symposium on
  Artificial Intelligence and Programming Languages, USA, August 15-17, 1977},
  pages 55--59. {ACM}, 1977.
\newblock \href {https://doi.org/10.1145/800228.806932}
  {\path{doi:10.1145/800228.806932}}.

\bibitem{bravetti21}
Mario Bravetti, Marco Carbone, Julien Lange, Nobuko Yoshida, and Gianluigi
  Zavattaro.
\newblock A sound algorithm for asynchronous session subtyping and its
  implementation.
\newblock {\em Log. Methods Comput. Sci.}, 17(1), 2021.
\newblock URL: \url{https://lmcs.episciences.org/7238}.

\bibitem{bravetti2008}
Mario Bravetti, Ivan Lanese, and Gianluigi Zavattaro.
\newblock Contract-driven implementation of choreographies.
\newblock In Christos Kaklamanis and Flemming Nielson, editors, {\em
  Trustworthy Global Computing, 4th International Symposium, {TGC} 2008,
  Barcelona, Spain, November 3-4, 2008, Revised Selected Papers}, volume 5474
  of {\em Lecture Notes in Computer Science}, pages 1--18. Springer, 2008.
\newblock \href {https://doi.org/10.1007/978-3-642-00945-7\_1}
  {\path{doi:10.1007/978-3-642-00945-7\_1}}.

\bibitem{carbone2012a}
Marco Carbone, Kohei Honda, and Nobuko Yoshida.
\newblock Structured {{Communication-Centered Programming}} for {{Web
  Services}}.
\newblock {\em ACM Transactions on Programming Languages and Systems},
  34(2):1--78, June 2012.
\newblock \href {https://doi.org/10.1145/2220365.2220367}
  {\path{doi:10.1145/2220365.2220367}}.

\bibitem{carbone2013}
Marco Carbone and Fabrizio Montesi.
\newblock Deadlock-freedom-by-design: Multiparty asynchronous global
  programming.
\newblock In Roberto Giacobazzi and Radhia Cousot, editors, {\em The 40th
  Annual {{ACM SIGPLAN-SIGACT}} Symposium on Principles of Programming
  Languages, {{POPL}} '13, Rome, Italy - January 23 - 25, 2013}, pages
  263--274. {ACM}, 2013.
\newblock \href {https://doi.org/10.1145/2429069.2429101}
  {\path{doi:10.1145/2429069.2429101}}.

\bibitem{CMMP22}
Lu{\'{\i}}s Cruz{-}Filipe, Anne Madsen, Fabrizio Montesi, and Marco Peressotti.
\newblock Modular choreographies: Bridging alice and bob notation to java.
\newblock In Gokila Dorai, Maurizio Gabbrielli, Giulio Manzonetto, Aomar
  Osmani, Marco Prandini, Gianluigi Zavattaro, and Olaf Zimmermann, editors,
  {\em Joint Post-proceedings of the Third and Fourth International Conference
  on Microservices, Microservices 2020/2022, May 10-12, 2022, Paris, France},
  volume 111 of {\em OASIcs}, pages 3:1--3:18. Schloss Dagstuhl -
  Leibniz-Zentrum f{\"{u}}r Informatik, 2022.
\newblock URL: \url{https://doi.org/10.4230/OASIcs.Microservices.2020-2022.3},
  \href {https://doi.org/10.4230/OASICS.MICROSERVICES.2020-2022.3}
  {\path{doi:10.4230/OASICS.MICROSERVICES.2020-2022.3}}.

\bibitem{cruz-filipe2017a}
Lu{\'i}s {Cruz-Filipe} and Fabrizio Montesi.
\newblock On {{Asynchrony}} and {{Choreographies}}.
\newblock {\em Electronic Proceedings in Theoretical Computer Science},
  261:76--90, November 2017.
\newblock \href {https://doi.org/10.4204/EPTCS.261.8}
  {\path{doi:10.4204/EPTCS.261.8}}.

\bibitem{cruz-filipe2017}
Lu{\'{\i}}s Cruz{-}Filipe and Fabrizio Montesi.
\newblock Procedural choreographic programming.
\newblock In Ahmed Bouajjani and Alexandra Silva, editors, {\em Formal
  Techniques for Distributed Objects, Components, and Systems - 37th {IFIP}
  {WG} 6.1 International Conference, {FORTE} 2017, Held as Part of the 12th
  International Federated Conference on Distributed Computing Techniques,
  DisCoTec 2017, Neuch{\^{a}}tel, Switzerland, June 19-22, 2017, Proceedings},
  volume 10321 of {\em Lecture Notes in Computer Science}, pages 92--107.
  Springer, 2017.
\newblock \href {https://doi.org/10.1007/978-3-319-60225-7\_7}
  {\path{doi:10.1007/978-3-319-60225-7\_7}}.

\bibitem{cruz-filipe2018}
Lu{\'i}s {Cruz-Filipe}, Fabrizio Montesi, and Marco Peressotti.
\newblock Communications in choreographies, revisited.
\newblock In Hisham~M. Haddad, Roger~L. Wainwright, and Richard Chbeir,
  editors, {\em Proceedings of the 33rd {{Annual ACM Symposium}} on {{Applied
  Computing}}, {{SAC}} 2018, {{Pau}}, {{France}}, {{April}} 09-13, 2018}, pages
  1248--1255. {ACM}, 2018.
\newblock \href {https://doi.org/10.1145/3167132.3167267}
  {\path{doi:10.1145/3167132.3167267}}.

\bibitem{CMP23}
Lu{\'{\i}}s Cruz{-}Filipe, Fabrizio Montesi, and Marco Peressotti.
\newblock A formal theory of choreographic programming.
\newblock {\em J. Autom. Reason.}, 67(2):21, 2023.
\newblock URL: \url{https://doi.org/10.1007/s10817-023-09665-3}, \href
  {https://doi.org/10.1007/S10817-023-09665-3}
  {\path{doi:10.1007/S10817-023-09665-3}}.

\bibitem{cutner22}
Zak Cutner, Nobuko Yoshida, and Martin Vassor.
\newblock Deadlock-free asynchronous message reordering in rust with multiparty
  session types.
\newblock In Jaejin Lee, Kunal Agrawal, and Michael~F. Spear, editors, {\em
  PPoPP '22: 27th {ACM} {SIGPLAN} Symposium on Principles and Practice of
  Parallel Programming, Seoul, Republic of Korea, April 2 - 6, 2022}, pages
  246--261. {ACM}, 2022.
\newblock \href {https://doi.org/10.1145/3503221.3508404}
  {\path{doi:10.1145/3503221.3508404}}.

\bibitem{DH12}
Romain Demangeon and Kohei Honda.
\newblock Nested protocols in session types.
\newblock In Maciej Koutny and Irek Ulidowski, editors, {\em {CONCUR} 2012 -
  Concurrency Theory - 23rd International Conference, {CONCUR} 2012, Newcastle
  upon Tyne, UK, September 4-7, 2012. Proceedings}, volume 7454 of {\em Lecture
  Notes in Computer Science}, pages 272--286. Springer, 2012.
\newblock \href {https://doi.org/10.1007/978-3-642-32940-1\_20}
  {\path{doi:10.1007/978-3-642-32940-1\_20}}.

\bibitem{giallorenzo2023}
Saverio Giallorenzo, Fabrizio Montesi, and Marco Peressotti.
\newblock Choral: Object-oriented choreographic programming.
\newblock {\em {ACM} Trans. Program. Lang. Syst.}, 46(1):1:1--1:59, 2024.
\newblock \href {https://doi.org/10.1145/3632398} {\path{doi:10.1145/3632398}}.

\bibitem{GHM24}
Eva Graversen, Andrew~K. Hirsch, and Fabrizio Montesi.
\newblock Alice or bob?: Process polymorphism in choreographies.
\newblock {\em J. Funct. Program.}, 34, 2024.
\newblock URL: \url{https://doi.org/10.1017/s0956796823000114}, \href
  {https://doi.org/10.1017/S0956796823000114}
  {\path{doi:10.1017/S0956796823000114}}.

\bibitem{hirsch2022}
Andrew~K. Hirsch and Deepak Garg.
\newblock Pirouette: Higher-order typed functional choreographies.
\newblock {\em Proc. ACM Program. Lang.}, 6(POPL):1--27, 2022.
\newblock \href {https://doi.org/10.1145/3498684} {\path{doi:10.1145/3498684}}.

\bibitem{honda2016}
Kohei Honda, Nobuko Yoshida, and Marco Carbone.
\newblock Multiparty asynchronous session types.
\newblock {\em J. {ACM}}, 63(1):9:1--9:67, 2016.
\newblock \href {https://doi.org/10.1145/2827695} {\path{doi:10.1145/2827695}}.

\bibitem{JB22}
Sung{-}Shik Jongmans and Petra van~den Bos.
\newblock A predicate transformer for choreographies - computing preconditions
  in choreographic programming.
\newblock In Ilya Sergey, editor, {\em Programming Languages and Systems - 31st
  European Symposium on Programming, {ESOP} 2022, Held as Part of the European
  Joint Conferences on Theory and Practice of Software, {ETAPS} 2022, Munich,
  Germany, April 2-7, 2022, Proceedings}, volume 13240 of {\em Lecture Notes in
  Computer Science}, pages 520--547. Springer, 2022.
\newblock \href {https://doi.org/10.1007/978-3-030-99336-8\_19}
  {\path{doi:10.1007/978-3-030-99336-8\_19}}.

\bibitem{KSZK23}
Shun Kashiwa, Gan Shen, Soroush Zare, and Lindsey Kuper.
\newblock Portable, efficient, and practical library-level choreographic
  programming.
\newblock {\em CoRR}, abs/2311.11472, 2023.
\newblock URL: \url{https://doi.org/10.48550/arXiv.2311.11472}, \href
  {https://arxiv.org/abs/2311.11472} {\path{arXiv:2311.11472}}, \href
  {https://doi.org/10.48550/ARXIV.2311.11472}
  {\path{doi:10.48550/ARXIV.2311.11472}}.

\bibitem{lanese2008}
Ivan Lanese, Claudio Guidi, Fabrizio Montesi, and Gianluigi Zavattaro.
\newblock Bridging the gap between interaction- and process-oriented
  choreographies.
\newblock In Antonio Cerone and Stefan Gruner, editors, {\em Sixth {{IEEE}}
  International Conference on Software Engineering and Formal Methods, {{SEFM}}
  2008, Cape Town, South Africa, 10-14 November 2008}, pages 323--332. {IEEE
  Computer Society}, 2008.
\newblock \href {https://doi.org/10.1109/SEFM.2008.11}
  {\path{doi:10.1109/SEFM.2008.11}}.

\bibitem{lugovic2023}
Lovro Lugović and Fabrizio Montesi.
\newblock Real-world choreographic programming: Full-duplex asynchrony and
  interoperability.
\newblock {\em The Art, Science, and Engineering of Programming}, 8(2), October
  2023.
\newblock URL:
  \url{http://dx.doi.org/10.22152/programming-journal.org/2024/8/8}, \href
  {https://doi.org/10.22152/programming-journal.org/2024/8/8}
  {\path{doi:10.22152/programming-journal.org/2024/8/8}}.

\bibitem{merro2004}
Massimo Merro and Davide Sangiorgi.
\newblock On asynchrony in name-passing calculi.
\newblock {\em Mathematical Structures in Computer Science}, 14(5):715--767,
  October 2004.
\newblock \href {https://doi.org/10.1017/S0960129504004323}
  {\path{doi:10.1017/S0960129504004323}}.

\bibitem{montesi2013phd}
Fabrizio Montesi.
\newblock {\em Choreographic {P}rogramming}.
\newblock Ph.{D}. thesis, IT University of Copenhagen, 2013.
\newblock
  \url{https://www.fabriziomontesi.com/files/choreographic-programming.pdf}.

\bibitem{montesi2023}
Fabrizio Montesi.
\newblock {\em Introduction to {{Choreographies}}}.
\newblock {Cambridge University Press}, {Cambridge}, 2023.

\bibitem{pohjola2022}
Johannes~Aman Pohjola, Alejandro {G{\'o}mez-Londo{\~n}o}, James Shaker, and
  Michael Norrish.
\newblock Kalas: {{A Verified}}, {{End-To-End Compiler}} for a {{Choreographic
  Language}}.
\newblock In June Andronick and Leonardo de~Moura, editors, {\em 13th
  {{International Conference}} on {{Interactive Theorem Proving}}, {{ITP}}
  2022, {{August}} 7-10, 2022, {{Haifa}}, {{Israel}}}, volume 237 of {\em
  {{LIPIcs}}}, pages 27:1--27:18. {Schloss Dagstuhl - Leibniz-Zentrum f\"ur
  Informatik}, 2022.
\newblock \href {https://doi.org/10.4230/LIPIcs.ITP.2022.27}
  {\path{doi:10.4230/LIPIcs.ITP.2022.27}}.

\bibitem{DGGLM16}
Mila~Dalla Preda, Maurizio Gabbrielli, Saverio Giallorenzo, Ivan Lanese, and
  Jacopo Mauro.
\newblock Dynamic choreographies: Theory and implementation.
\newblock {\em Log. Methods Comput. Sci.}, 13(2), 2017.
\newblock \href {https://doi.org/10.23638/LMCS-13(2:1)2017}
  {\path{doi:10.23638/LMCS-13(2:1)2017}}.

\bibitem{qiu2007a}
Zongyan Qiu, Xiangpeng Zhao, Chao Cai, and Hongli Yang.
\newblock Towards the theoretical foundation of choreography.
\newblock In {\em Proceedings of the 16th International Conference on {{World
  Wide Web}}}, pages 973--982, {Banff Alberta Canada}, May 2007. {ACM}.
\newblock \href {https://doi.org/10.1145/1242572.1242704}
  {\path{doi:10.1145/1242572.1242704}}.

\bibitem{scharf2006}
Michael Scharf and Sebastian Kiesel.
\newblock Head-of-line {{Blocking}} in {{TCP}} and {{SCTP}}: {{Analysis}} and
  {{Measurements}}.
\newblock In {\em Proceedings of the {{Global Telecommunications Conference}},
  2006. {{GLOBECOM}} '06, {{San Francisco}}, {{CA}}, {{USA}}, 27 {{November}} -
  1 {{December}} 2006}. {IEEE}, 2006.
\newblock \href {https://doi.org/10.1109/GLOCOM.2006.333}
  {\path{doi:10.1109/GLOCOM.2006.333}}.

\bibitem{shen2023a}
Gan Shen, Shun Kashiwa, and Lindsey Kuper.
\newblock Haschor: Functional choreographic programming for all (functional
  pearl).
\newblock {\em Proc. {ACM} Program. Lang.}, 7({ICFP}):541--565, 2023.
\newblock \href {https://doi.org/10.1145/3607849} {\path{doi:10.1145/3607849}}.

\bibitem{VHEZ21}
Malte Viering, Raymond Hu, Patrick Eugster, and Lukasz Ziarek.
\newblock A multiparty session typing discipline for fault-tolerant
  event-driven distributed programming.
\newblock {\em Proc. {ACM} Program. Lang.}, 5({OOPSLA}):1--30, 2021.
\newblock \href {https://doi.org/10.1145/3485501} {\path{doi:10.1145/3485501}}.

\bibitem{wang21}
Stephanie Wang, Eric Liang, Edward Oakes, Benjamin Hindman, Frank~Sifei Luan,
  Audrey Cheng, and Ion Stoica.
\newblock Ownership: {A} distributed futures system for fine-grained tasks.
\newblock In James Mickens and Renata Teixeira, editors, {\em 18th {USENIX}
  Symposium on Networked Systems Design and Implementation, {NSDI} 2021, April
  12-14, 2021}, pages 671--686. {USENIX} Association, 2021.
\newblock URL:
  \url{https://www.usenix.org/conference/nsdi21/presentation/cheng}.

\end{thebibliography}

\clearpage

\appendix

\section{Well-formedness}\label{sec:well-formedness}
\begin{center}

    \begin{align*}
        &\fv(0) = \emptyset\\
        &\fv(\{\,C\,\};\,C') = \fv(C) \cup \fv(C')\\
        &\fv(\CCom{\p}{e}{\q}{x}{l,t};\,C') = \fv(e) \cup \fv(C') \setminus \{\q.x\}\\
        &\fv(\CPure{x}{\p}{e}{l,t};\,C') = \fv(e) \cup \fv(C') \setminus \{\p.x\}\\
        &\fv(\CIf{e}{\p}{C_1}{C_2}{l,t};\,C') =\\ 
        &\quad \fv(e) \cup \fv(C_1) \cup \fv(C_2) \cup \fv(C')\\
        &\fv(\CCall{X}{\vx\p,\vx a}{l,t};\,C') = \{ \p.x\ |\ \p.x \in \vx a \} \cup \fv(C')\\
        &\fv(\CCalling{\vx\p}{X}{\vx\q, \vx a}{C}{l,t}) = \{ \p.x\ |\ \p.x \in \vx a \} \cup \fv(C')\\
        &\fv(I;\,C) = \fv(C) \ \hbox{otherwise.}\\
        \end{align*}

\begin{prooftree}
    \noLine
    \AxiomC{$\pn(C) \subseteq \dom(\Sigma)$ \qquad $\pn(C) \subseteq \dom(K)$ \qquad $\fv(C) = \emptyset$}
    \UnaryInfC{ $\keys(C)$ distinct \qquad $\forall (l,t) \in \keys(C),\, t \ne \Tok$}
    \noLine
    \UnaryInfC{$\forall I_1,I_2 \in \stats(C)$, if $\key(I_1) \prec \key(I_2)$ then $I_1 = l_1,t_1:\vx\q.X(\vx\p,\vx a)\{C'\}$ and $I_2 \in \stats(C')$}
    \noLine
    \UnaryInfC{$(X(\vx\p, \vx{\p.x}) = C) \wf$ for each $(X(\vx\p, \vx{\p.x}) = C) \in \mathscr C$ \qquad $\langle I, K\rangle \wf$ for each $I \in \stats(C)$}
    \RightLabel{\textsc{C-WF}}
    \UnaryInfC{ $\langle\mathscr C, C, \Sigma, K\rangle \wf$ }
\end{prooftree}

\begin{prooftree}
    \AxiomC{$\exists! v,\,(l,\tau,v) \in K(\q)$}
    \RightLabel{\textsc{C-WF-Recv}}
    \UnaryInfC{ $\langle \CRecv{\p}{\q}{x}{\Line,\tau},\, {K} \rangle \wf$ }
\end{prooftree}

\begin{prooftree}
    \AxiomC{$(l,\tau,L) \in K(\q)$}
    \RightLabel{\textsc{C-WF-OnSelect}}
    \UnaryInfC{ $\langle \CRecvSelect{\p}{\q}{L}{\Line,\tau},\, {K} \rangle \wf$ }
\end{prooftree}

\begin{prooftree}
    \AxiomC{}
    \RightLabel{\textsc{C-WF-Compute}}
    \UnaryInfC{ $\langle \CPure{x}{\p}{e}{\Line,\tau},\, {K} \rangle \wf$ }
\end{prooftree}

\begin{prooftree}
    \AxiomC{$C_1,C_2$ contain no runtime terms}
    \RightLabel{\textsc{C-WF-If}}
    \UnaryInfC{ $\langle \CIf{e}{\p}{C_1}{C_2}{\Line,\tau},\, {K} \rangle \wf$ }
\end{prooftree}

\begin{prooftree}
    \AxiomC{ }
    \RightLabel{\textsc{C-WF-Block}}
    \UnaryInfC{ $\langle \Block{C_1},\, {K} \rangle \wf$ }
\end{prooftree}

\begin{prooftree}
    \noLine
    \AxiomC{ $(X(\q_1,\dots,\q_n, \q^1.x_1, \dots, q^m.x_m) = C) \in \mathscr C$}
    \UnaryInfC{ $\p_1,\dots,\p_n$ distinct \qquad $\forall i \le n, j \le m,$ if $\pn(a_j) = \p_i$ then $\q^j = \q_i$ }
    \RightLabel{\textsc{C-WF-Call}}
    \UnaryInfC{ $\langle \CCall{X}{\p_1,\dots,\p_n, a_1,\dots,a_m}{\Line,\tau},\, {K} \rangle \wf$ }
\end{prooftree}

\end{center}

\section{EPP Theorem}\label{sec:epp-theorem-proof}

\changed{}{Proving the EPP Theorem requires first establishing several basic properties about choreographies. \Cref{lem:subst-local,lem:subst-branching} are properties of substitution: \Cref{lem:subst-local} states that local variables at a process $\p$ are projected to local variables at that same location in the network and \Cref{lem:subst-branching} states that substitution preserves the $(\Extends)$ relation. \Cref{lem:proj-seq} states that projection preserves sequential composition---a modified version of Lemma 8.17 from Montesi~\cite{montesi2023}.}

\begin{lemma}\label{lem:subst-local}
  $$\EPP{C[\p.x \mapsto v]}{\r} =
  \begin{cases}
    \EPP{C}{\p}[x \mapsto v] & \hbox{if $\r = \p$}\\
    \EPP{C[\p.x \mapsto v]}{\q} = \EPP{C}{\q} & \hbox{otherwise.}
  \end{cases}
  $$
\end{lemma}

\begin{lemma}\label{lem:subst-branching}
  $P \Extends Q$ implies $P[x \mapsto v] \Extends Q[x \mapsto v]$.
\end{lemma}

\begin{lemma}\label{lem:proj-seq}
  $\EPP{I;\,C}{\q} = \EPP{I;\,0}{\q}\fatsemi\,\EPP{C}{\q}$ if $I$ does not have the form $\CSelect{\p}{\q}{L}{l,t}$ or $\CRecvSelect{\p}{\q}{L}{l,t}$.
\end{lemma}

\changed{}{We can now prove the EPP Theorem, which follows from standard proof techniques. The only novelty is (1) our use of substitution instead of local state to bind variables to values; (2) our use of well-formedness; and (3) a large number of extra cases in the soundness direction, generated by the \textsc{P-Delay} rule. Fortunately, all these extra cases are dispatched easily. We include a proof sketch below to show how this is done.}

  \epp*
  \begin{proof}[Completeness, sketch] The proof proceeds by induction on the derivation $\mathcal D$ that produces $\CCfg{C}{\Sigma}{K} \CStep{\p} \CCfg{C'}{\Sigma'}{K'}$. In most cases it suffices to let $N' = \EPP{C'}{}$. In the case of \textsc{C-If}, we find a network $N'$ such that $N' \ne \EPP{C'}{}$ but $N' \Extends \EPP{C'}{}$.

Because $O_3$ uses scoped variables, the proof requires \Cref{lem:subst-local}. The \textsc{C-Delay} rule is also slightly novel: if $I$ is not a selection at $\q$ then completeness requires \Cref{lem:proj-seq} and an application of \textsc{P-Delay}.
\end{proof}
\begin{proof}[Soundness, sketch]

  The proof proceeds by induction on the structure of the choreography $C$. The base case, $C \equiv 0$, is trivial. Otherwise, $C \equiv (I;\,C')$ and there is a distinct case for each instruction $I$. The key novelty in this proof is handling \textsc{P-Delay} and the frequent use of well-formedness to guarantee that $K$ does not contain certain messages. We consider two representative cases.

    \vspace{0.4cm}

    \noindent\textbf{Case 1.} Let $C = \CCom{\p}{e}{\q}{x}{l,\tau};\, C''$. Then, by the definition of EPP, $N$ has the form
    $$N = \p[\PSend{\q}{l}{\tau}{e};\, P]\ |\ \q[\PRecv{x}{l}{\tau};\,Q]\ |\ (N \setminus \p,\q).$$
    There are three sub-cases.

    \vspace{0.4cm}

    \noindent\emph{Case 1.1.} Assume $\r = \p$. In standard choreography models, this case can only proceed by \textsc{P-Send}. In our model, it could also proceed by \textsc{P-Delay}. Hence there are two sub-cases:

    \noindent\emph{Case 1.1.1.} (\textsc{P-Send}) Satisfied by letting $C' = \CRecv{\p}{\q}{x}{l,\tau};\,C''$.

    \noindent\emph{Case 1.1.2.} (\textsc{P-Delay})
    By the induction hypothesis, there exists $C'''$ such that $N''' \Extends \EPP{C'''}{}$.
    The case is then satisfied by letting $C' = \CCom{\p}{e}{\q}{x}{l,\tau};\, C'''$.

    \vspace{0.4cm}

    \noindent\emph{Case 1.2.} Assume $\r = \q$. Again, we consider the two rules by which the case could proceed:

    \noindent\emph{Case 1.2.1.} (\textsc{P-Recv}) We must show that it is impossible for $\q$ to receive a message in $\CCfg{N}{\Sigma}{K}$. Since $C$ is well-formed, it cannot contain a communication-in-progress term $\CRecv{\p}{\q}{x}{l,\tau}$ with the integrity key $(l,\tau)$. Hence $M$ does not contain any messages of the form $(l,\tau,v)$.

    \noindent\emph{Case 1.2.2.} (\textsc{P-Delay}) This case proceeds similarly to the previous \textsc{P-Delay} case.

    \vspace{0.4cm}

    \noindent\emph{Case 1.3.} Assume $\r \notin \{\p,\q\}$. Follows from the induction hypothesis, as in Case 1.1.2.




    \vspace{0.4cm}

    \noindent\textbf{Case 2.} Let $C = \CIf{e}{\p}{C_1}{C_2}{l,\tau};\, C_3$. Then $N$ has the form
    $$N = \p[\PIf{e}{P_1}{P_2};\,P_3]\ |\ \left(\prod_{q_i \in \vx\q} \q_i[\PBranch{l_j,\tau_j,L_j}{Q_{i,j}}{j \in \mathcal I};\,Q_i] \right)\ |\ (N \setminus \p, \vx\q).$$
    Consider the case where $\r = \q_i \in \vx\q$ and $\CCfg{N}{\Sigma}{K} \PStep{\q} \CCfg{N'}{\Sigma'}{K'}$ proceeds by \textsc{P-OnSelect}. That is,
    $$\q_i[\PBranch{l_j,\tau_j,L_j}{Q_{i,j}}{j \in \mathcal I};\, Q_i] \PStep{\q_i} \q_i[Q_{i,k};\, Q_i]$$
    for some $k \in \mathcal J$. We must show that this case is impossible. Notice that the step can only occur if $(l,\tau_k,L_k) \in K(\q_i)$. Such a message can only occur if $\CRecvSelect{\p}{\q}{L_k}{l,\tau_k}$ occurs in $C_1,C_2,$ or $C_3$. Because $C$ is well-formed, $C_1$ and $C_2$ do not contain runtime terms; hence the term could only occur in $C_3$. But then $Q_i$ would contain a branch $(l,\tau_k,L_k) \Rightarrow Q'$ and $\keys(N|_{\q_i})$ would not be distinct; a contradiction of \Cref{lem:branching-keys} (5).
\end{proof}

\end{document}